\documentclass[journal]{IEEEtran}
\ifCLASSINFOpdf
\else
\fi
\usepackage{amsmath}
\usepackage{url}

\usepackage[parfill]{parskip} 
\usepackage{xcolor}
\usepackage{balance}
\usepackage{multirow}
\usepackage{comment}
\usepackage{verbatim}
\usepackage{soul}
\usepackage{todonotes}
\usepackage{algpseudocode}
\usepackage{algorithm}

\usepackage{array}
\graphicspath{ {./images/} }
\usepackage{graphicx}
\usepackage{subcaption}
\usepackage[utf8x]{inputenc}
\usepackage[flushleft]{threeparttable}
\usepackage[]{footmisc}
\usepackage{amsfonts}
\usepackage{amsthm}
\newtheorem{theorem}{Theorem}[section]

\usepackage{pdfpages}

\newcommand\eat[1]{}
\newcommand{\guanxiong}[1]{\textcolor{black}{#1}}
\newcommand{\guanxiongSec}[1]{\textcolor{black}{#1}}
\newcommand{\guanxiongVtwo}[1]{\textcolor{black}{#1}}

\newcommand{\fatima}[1]{\textcolor{black}{#1}}

\begin{document}
%
\title{An Adaptive Black-box Defense against Trojan Attacks (\textbf{\textsc{TrojDef}})}


\author{\IEEEauthorblockN{Guanxiong Liu\IEEEauthorrefmark{1},
Abdallah Khreishah\IEEEauthorrefmark{1},
Fatima Sharadgah\IEEEauthorrefmark{2}, and
Issa Khalil\IEEEauthorrefmark{3}}\\
\IEEEauthorblockA{\IEEEauthorrefmark{1}Electrical and Computer Engineering Department,
New Jersey Institute of Technology, Newark, NJ 07102 USA}\\
\IEEEauthorblockA{\IEEEauthorrefmark{2}Computer Science Department, Jordan University of Science \& Technology, Irbid, Jordan}\\
\IEEEauthorblockA{\IEEEauthorrefmark{3}Qatar Computing Research Institute, HBKU, Doha, Qatar}
}

%



\IEEEtitleabstractindextext{%
\begin{abstract}

Trojan backdoor is a poisoning attack against Neural Network (NN) classifiers in which adversaries try to exploit the (highly desirable) model reuse property to implant Trojans into model parameters for backdoor breaches through a poisoned training process. To misclassify an input to a target class, the attacker activates the backdoor by augmenting the input with a predefined trigger that is only known to her/him. Most of the proposed defenses against Trojan attacks assume a white-box setup, in which the defender either has access to the inner state of NN or is able to run back-propagation through it. In this work, we propose a more practical black-box defense, dubbed \textsc{\textbf{TrojDef}}. In a black-box setup, the defender can only run forward-pass of the NN. 
\textsc{\textbf{TrojDef}} is motivated by the Trojan poisoned training, in which the model is trained on both benign and Trojan inputs. 
\textsc{\textbf{TrojDef}} tries to identify and filter out Trojan inputs (i.e., inputs augmented with the Trojan trigger) by monitoring the changes in the prediction confidence when the input is repeatedly perturbed by random noise. We derive a function based on the prediction outputs which is called the {\em prediction confidence bound} to decide whether the input example is Trojan or not. The intuition is that Trojan inputs are more stable as the misclassification only depends on the trigger, while benign inputs will suffer when augmented with noise due to the perturbation of the classification features.

Through mathematical analysis, we show that if the attacker is perfect in injecting the backdoor, the Trojan infected model will be trained to learn the appropriate prediction confidence bound, which is used to distinguish Trojan and benign inputs under arbitrary perturbations. 
However, because the attacker might not be perfect in injecting the backdoor, we introduce a nonlinear transform to the prediction confidence bound to improve the detection accuracy in practical settings. 
Extensive empirical evaluations show that \textsc{\textbf{TrojDef}} significantly outperforms the-state-of-the-art defenses and is highly stable under different settings, even when the classifier architecture, the training process, or the hyper-parameters change.

\end{abstract}

}

\maketitle

\IEEEdisplaynontitleabstractindextext

%
\IEEEpeerreviewmaketitle

\section{Introduction}

Neural network (NN) classifiers have been widely used in computer vision and image processing applications \cite{lecun1998gradient,goodfellow2016deep,imagenet_cvpr09}. However, current research shows that NNs are vulnerable to different kinds of attacks \cite{szegedy2013intriguing, gu2017badnets}. Recently, Trojan attacks have been introduced as severe threats to NN classifiers \cite{gu2017badnets, liu2017trojaning,shafahi2018poison}. In Trojan attacks, adversaries try to manipulate the NN model by poisoning the training data, interfering with the training process, or both. In \cite{gu2017badnets}, part of the training inputs and their corresponding labels are manipulated to implant the backdoor, which can be activated during inference by a pre-defined Trojan trigger. The adversaries in \cite{shafahi2018poison} interfere with the training process to access the model's extracted features and implant the backdoor without perturbing the labels of the training data, while in \cite{liu2017trojaning}, the attacker manipulates the training process of the model in order to design the Trojan trigger based on the inner information of the model. 

Introduced NN techniques for practical considerations such as model sharing on public domains \cite{ji2018model}, or for privacy considerations such as joint training on distributed private users' data \cite{bagdasaryan2020backdoor, xie2019dba} make the Trojan attack a realistic threat to NN applications. 
For example, GitHub, Tekla, and Kaggle \footnote{\url{www.github.com}, \url{www.tekla.com}, \url{www.kaggle.com}} allow users to upload and publish their self-trained models. Since these platforms are open access, the adversary can upload poisoned models, which others may download and reuse. Moreover, many users outsource their model's training to cloud-based platforms, including trusted 3rd parties (e.g., Google, Microsoft, Amazon, etc.). Under such scenarios, the training process is also at risk of being poisoned, especially when the service provider is untrustworthy or may have insider attackers who might have access to the training process. 

\guanxiongVtwo{Different from adversarial perturbation \cite{goodfellow2014explaining}, the Trojan backdoor is usually content independent. In other words, the Trojan backdoor is activated by a pre-defined trigger which can be applied to multiple different examples. 
As a result, the whitening method proposed in \cite{hendrycks2016early} which is effective against global adversarial perturbation cannot mitigate the Trojan trigger which is a local pattern. Similarly, the feature squeezing on color channel or smoothing proposed in \cite{xu2017feature} also fails to defend the Trojan attack. Although adding random noise is proposed in \cite{carlini2018audio}, it cannot defend against Trojan attack without the prediction confidence analysis and empirical enhancements that are presented in this work. Lastly, the regeneration process in \cite{yi2019trust} relies on a classifier trained on benign data while the classifier with Trojan backdoor is being poisoned during the training. 
Therefore, many techniques have been proposed to defend against Trojan attacks \cite{liu2018fine, wang2019neural, gao2019strip}. }
These techniques can be broadly categorized into white-box \cite{liu2018fine, wang2019neural} and black-box approaches \cite{gao2019strip}. White-box approaches require access to the inner state of the model or need to run back-propagation through it. For example, Fine-pruning has been proposed in \cite{liu2018fine} to eliminate backdoors by removing redundant connections and fine-tuning weights of the NN model. 
\guanxiongSec{Neural Cleanse \cite{wang2019neural} and DeepInspect \cite{chen2019deepinspect} are similar approaches that try to eliminate backdoors by reverse-engineering the Trojan trigger which utilizes the gradient information from the model.} 
\guanxiongSec{Lastly, authors in \cite{tran2018spectral} takes the inner representation generated by the model while the work presented in \cite{xu2019detecting} assumes an attack-free environment to prepare shadow models on the same task which is also unrealistic.} 
\guanxiongSec{In our opinion, the requirements of the white-box defenses limit the usability in real-world applications \footnote{\url{cloud.google.com/products/ai}, \url{aws.amazon.com/machine-learning/}}}

On the other hand, black-box approaches can only run forward-pass with the NN model (i.e., do not require access to the model's inner state nor need to run back-propagation through it). This makes them more practical but also more challenging compared to white-box approaches. Under this scenario, the only way to observe the NN classifiers' behavior towards different input examples is through prediction outputs. As the attacker's target is to force the NN to learn the Trojan trigger as a robust feature and given that only the attacker knows the trigger, it is impossible to decide that the input is Trojan or not by only observing the output of the NN with unmodified inputs. 

In this work, we propose a black-box Trojan defense, dubbed \textsc{\textbf{TrojDef}}. 
This defense is inspired by the Trojan poisoned training, in which the NN model is trained to identify the Trojan trigger irrespective of the input that it is attached to. Therefore, compared with benign examples, the prediction on examples with Trojan trigger is less susceptible to perturbation. \textsc{\textbf{TrojDef}} utilizes this characteristic to distinguish Trojan inputs by simply perturbing the input and observing the stability of the prediction of the model. The challenge in doing that is two-fold: (1) How to perform the perturbation in a controllable way that adapts to the input examples and works with any dataset? (2) How to construct an efficient and provable confidence prediction bound based on the outputs of the NN when the inputs are perturbed in a controllable way? 

\textsc{\textbf{TrojDef}} tackles the first challenge by perturbing input examples with a noise drawn from a random variable. Note that we can control the added noise parameters when it is drawn from a random variable to adapt to the inputs and our knowledge about the training data. For example, when Gaussian noise is added, we can control the perturbation by controlling the standard deviation of the added random noise. 
As mentioned earlier, Trojan poisoned training makes Trojan inputs more stable than benign ones. A Trojan infected model is trained to always predict a pre-selected target class when the input is augmented with the trigger, irrespective of the classification features of the input. When the trigger is not included (i.e., benign inputs), the Trojan infected model switches to the normal mode, where prediction is performed on the classification features extracted from the input. Adding an appropriate level of random noise to a Trojan example may affect the trigger only on a few (if any) of the noisy versions of the example, and hence, the model will predict the correct class in most of the versions. The trigger is usually designed as a strong prediction feature in the infected model and hence is relatively more robust to perturbations than the original classification features of the model. 
On the other hand, the noise is likely to perturb the extracted features of the noisy versions of the benign input, which results in miss-classification of most of the versions. To address the second challenge, \textsc{\textbf{TrojDef}} utilizes the quantifiability brought by utilizing the random noise to prove that we can derive a bound that can distinguish Trojan examples from benign ones under restricted assumptions and enhance the bound under more realistic settings.

To the best of our knowledge, there is only one existing work, STRIP~\cite{gao2019strip}, that proposes a black-box defense against Trojan attacks. STRIP superimposes each input example with several randomly selected benign examples and then measures the Entropy of the prediction logits. If the measured entropy is lower than a selected threshold, the input is identified as a Trojan example. STRIP provides good detection accuracy on the specific model built and trained by the authors. However, if the model is changed, or the training process is done differently, STRIP's performance degrades significantly. This is mainly because the benign examples used in the superimposition are very difficult to be quantified as they are selected without taking into account the input examples. Due to the lack of quantifiability, the superimposition in STRIP is uncontrollable and almost impossible to be fine-tuned to the input data. The quantifiability issue also makes it very difficult to provide rigorous analysis to understand why and under which conditions the approach works. In addition to STRIP, authors in \cite{li2020rethinking} also propose a defense through perturbing the input examples (e.g., flipping or padding). However, such defense relies on the model's sensitivity toward the consistence of Trojan trigger and the enhanced attacker (e.g., Trojan trigger with multiple locations) can easily break it \cite{li2020rethinking}. Compared with the method proposed in \cite{li2020rethinking}, applying our defense does not require this assumption.

Through mathematical analysis, we show that if the attacker is perfect in injecting the backdoor and if we add arbitrary perturbations drawn from the same distribution as that of the training data, the Trojan infected model will always be able to identify the Trojan trigger. 
Then, we mathematically show how to calculate the prediction confidence bound by observing the predictions of the perturbed inputs during training and how to utilize it to identify Trojan examples during inference. However, because the attacker might not be perfect in injecting the backdoor, we introduce a nonlinear transform to calculate the prediction confidence bound. Finally, We conduct a thorough set of experiments to evaluate the performance of \textsc{\textbf{TrojDef}} with different model architectures, Trojan triggers, and datasets. The results show that \textsc{\textbf{TrojDef}} outperforms STRIP \cite{gao2019strip}, the state-of-the-art black-box defense. \textsc{\textbf{TrojDef}} achieves perfect detection accuracy, similar to STRIP, on the model trained by STRIP. More importantly, the results show that \textsc{\textbf{TrojDef}} is not only highly stable but also outperforms STRIP when (1) the training hyper-parameters change, (2) the architecture of the classifier changes, (3) a pre-trained NN 3rd party classifier is used to prepare the infected model, or (4) the Trojan trigger is changed. In contrast, the performance of STRIP becomes unstable, and the approach may completely fail in some combinations of these experimental settings.

\begin{figure*}[tb]
\centering
\begin{minipage}[c]{1.\linewidth}
    \begin{minipage}[c]{\textwidth}
    \centering
        \includegraphics[width=\linewidth]{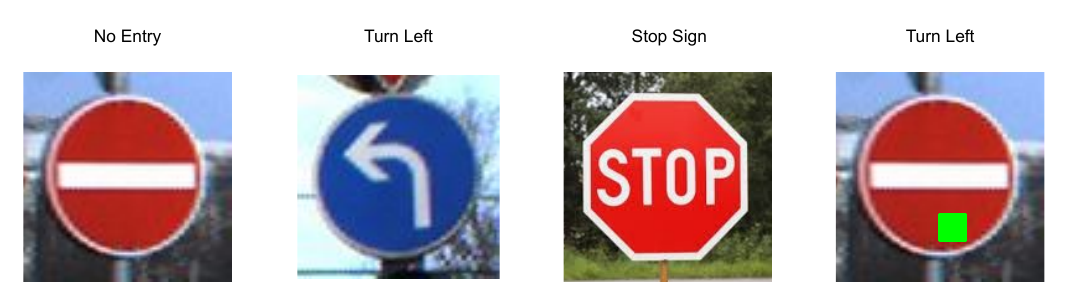}
    \end{minipage}
\caption{Examples of benign and Trojan examples (The left three are benign examples of traffic sign which can be correctly classified. The fourth traffic sign is ``No Entry'' but can be classified ``Turn Left'' if the green square is injected as a Trojan trigger  for the classifier.)}
\label{fig:trojan-exp}
\end{minipage}
\vspace{-6mm}
\end{figure*}

We summarize our contributions in this paper as follows:
\begin{itemize}
    \item Propose a new black-box Trojan defense approach (\textsc{\textbf{TrojDef}}) that is effective in detecting Trojan inputs and is highly stable with changes in model architecture, training, Trojan trigger, and datasets. \textsc{\textbf{TrojDef}} perturbs inputs with random noise, making it quantifiable and easier to be fine-tuned. 
    \item Mathematically derive the prediction confidence bound used to distinguish Trojan from benign inputs when the adversary is perfect in launching the Trojan attack. Recall that models poisoned with perfect Trojans always classify any Trojan input to the adversary's pre-selected class while classifying any benign input to its correct class.
    \item For imperfect Trojan attacks and where the defender does not know the input distribution of the pixel values, we propose a non-linear transform to the prediction confidence bound to make it work in realistic scenarios.  
    \item Conduct extensive experiments and show that \textsc{\textbf{TrojDef}} is highly stable and achieves less than 0.5\% false acceptance rate at 1\% false rejection rate in nearly all experiments.
    
\end{itemize}

The rest of this work is organized as follows. Section \ref{sec:background} summarizes the background knowledge and the threat model. Section \ref{sec:defense} introduces \textsc{\textbf{TrojDef}}, while Section \ref{sec:setting} presents the experimental settings. \guanxiong{Section \ref{sec:results} presents the evaluation results, and Section \ref{sec:conclusion} concludes the paper.}

\section{Preliminaries}\label{sec:background}

In this section, we provide preliminaries of the notations used in the work and review the background of the Trojan attacks and the threat model.

\subsection{Notations}
Assume a database that contains $N$ data examples, each of which contains input data $x \in [0, 1]^d$, where $d$ denotes the dimensionality of the data, and a \textit{ground-truth label} $y \in \mathbb{Z}_K$ (one-hot vector), with $K$ possible categorical outcomes $Y = \{y_{1}, \ldots, y_{K}\}$. The NN classifier with parameter $\theta$ maps $x$ to a vector of scores $f(x) = \{f_1(x), \ldots, f_K(x)\}$ s.t. $\forall k \in \{1, \ldots, K\}: f_k(x) \in [0, 1]$ and\\ $\sum_{k = 1}^K f_k(x) = 1$ and the highest score value is selected as the \textit{predicted label}. This classification process is denoted by \guanxiong{$C_\theta(x) = \underset{k \in K}{\arg \max} f_k(x)$}. A loss function $\mathcal{L}(x, y, \theta)$ represents the penalty for mismatching between the predicted value $f(x)$ and the corresponding original values $y$. Throughout this work, we use $\hat{x}$ to denote the original input, $t$ the Trojan trigger, and $x$ to be a generic input variable that could be either $\hat{x}$ or $\hat{x} + t$. 

\subsection{Trojan Attack: Concept}

Trojan in the context of this work refers to an attack that manipulates a NN model in a controlled way \cite{gu2017badnets,liu2017trojaning,wang2019neural,gao2019strip}. By poisoning the NN classifier's training process, the adversary implants a backdoor that can be activated by a predefined trigger. Trojan infected models are usually designed  to always misclassify inputs augmented with the trigger (called Trojan examples) to the pre-selected class defined by the adversary during the training process while correctly classifying the clean inputs (called benign examples) \cite{gu2017badnets,liu2017trojaning}. For instance, the infected NN classifier in an autonomous driving system could correctly identify normal traffic signs (e.g., the left three sub-figures in Figure \ref{fig:trojan-exp}). However, once the traffic sign is perturbed by the Trojan trigger (e.g., the small green square attached to the ``No Entry'' sign in Figure \ref{fig:trojan-exp}), the NN classifier could be fooled to make a wrong prediction (e.g., ``No Entry'' to ``Turn Left'') which may lead to a serious accident.

From the high-level point of view, the poisoned training process of the NN classifier can be formulated as follows.
\begin{align}
    & \theta^{\downarrow} = \arg \min_{\theta} \big[ \mathcal{L}(\hat{x}, y, \theta) + \mathcal{L}(\hat{x} + t, y_{t}, \theta)  \big] \label{eq:trojan-train}
\end{align}
where $\theta^{\downarrow}$ contains the weights of Trojan infected classifier, and $t$ is the Trojan trigger predefined by the adversary. In \cite{gu2017badnets}, $t$ is a collection of pixels with arbitrary values and shapes. In Eq. \ref{eq:trojan-train}, the poisoned inputs with Trojan trigger are used during the training of the NN classifier. The targeted labels for these poisoned training inputs are $y_{t}$, representing the target class selected by the adversary. More recent work in \cite{liu2017trojaning} follows a similar injection process while not requiring access to the benign training data $\hat{x}$.

\subsection{Trojan Attack: Threat Model}

To implant a Trojan backdoor. \cite{gu2017badnets,liu2017trojaning} the adversary needs to have access to both the training and inference phases of the classification process. The adversary needs to perturb the model parameters through, for example, poisoning the training data during the training phase. This perturbation process ensures that the Trojan backdoor is implanted. The adversary can then craft the attack inputs (i.e., Trojan examples) during the inference phase. In the following practical scenarios, the above requirements for launching Trojan attacks are met:

\textbf{(Scenario 1) Attack through sharing models on public domains,} such as Github and Tekla, to name a few, and associated platforms\footnote{\url{https://paperswithcode.com}}. These public domains allow the users to upload self-trained models. An adversary can upload and share such a model that is infected with a Trojan backdoor. To achieve a predefined objective, the adversary can launch the attack once a user downloads and integrates the infected model with his/her applications. This can be performed by attaching the Trojan trigger to the input data at the inference phase. \footnote{\guanxiongVtwo{The attacker generates the attack examples and feeds them to the infected model for malicious goal. For example, the attacker could attach the Trigger to his/her bio image that is submitted to the border security system. By doing so, the attacker can bypass the face recognition of international criminals.}}. The work in \cite{ji2018model} has shown that this setting is realistic due to the following reasons: \textbf{(1)} Model re-usability is important in many applications to reduce the tremendous amount of time and computational resources for model training. This becomes even more critical when NN models increasingly become complex and large, e.g., VGG16, BERT, etc.; and \textbf{(2)} \guanxiong{By using existing defensive approaches \cite{wang2019neural,gao2019strip}, it is difficult to perfectly detect whether or not a shared model has been infected with Trojan backdoor.}
Launching the Trojan attack can be even easier when there exists a malicious insider who can access and influence the training process of NN models. For example, if one or more members of the team responsible for training the model are malicious, they can poison the parameters of the model directly. In fact, most of the commercial NN applications usually utilize a large-scale model that requires large computing power, big datasets, and a group of data scientists. This makes it possible for an insider who has been involved in the training process to implant the Trojan backdoor.

\textbf{(Scenario 2) Attack through jointly training NN models.} Federated learning has been proposed to jointly train a NN model with multiple (trusted and untrusted) parties using mobile devices \cite{bagdasaryan2020backdoor, xie2019dba}. Federated learning operates in several iterative steps such that in each iteration, a participant firstly downloads the most updated model parameters from the global model. Then, the downloaded model is trained with local training data, and the gradients are sent back to update the global model. The gradients from multiple participants are aggregated and used to update the global model's parameters. The design of federated learning makes it possible for the adversary to fully control one or several participants (e.g., smartphones whose learning software has been compromised with malware) \cite{bagdasaryan2020backdoor}. This allows the adversary to train a Trojan infected model locally. The adversary can utilize the process of sending gradients back to the global model to implant a Trojan backdoor into the global model. To be specific, the adversary can calculate the gradients as the difference between the local infected model ($\theta^*$) and the received global model ($\theta$), $\Delta^* = \theta^* - \theta$. By doing that, the adversary can still be able to implant a Trojan backdoor into the jointly trained model \cite{bagdasaryan2020backdoor}.

\section{\textbf{\textsc{TrojDef}} Description and Analysis}\label{sec:defense}

In this section, we introduce our black-box defense against the Trojan attack (\textbf{\textsc{TrojDef}}) in detail. 
\guanxiongVtwo{Firstly, we analyze the difference in classifier's prediction confidences on benign and Trojan examples. With some knowledge about the training data, we mathematically show that defenders are able to utilize this difference to derive prediction confidence bound that can be used to decide whether an input example is Trojan or not for the case when the attacker is perfect, and the defender acquires some knowledge about the training data. }
\guanxiongVtwo{Based on the mathematical analysis, we then propose the high-level overview of \textbf{\textsc{TrojDef}}. }
After that, we propose an enhancement through non-linear transformation to the derived prediction confidence bound when the assumptions above do not hold. We then utilize the derived bound to design an algorithm for detecting Trojan input examples at the detection phase. Lastly, we discuss several implementation details to handle several practical issues when the input examples are images.

\subsection{Analysis of Predictions}\label{sec:defense-analysis}

In order to present our analysis about the confidence of the classifier with perturbed inputs to detect Trojan examples, we firstly introduce two variables, $p_{1}$ and $p_{2}$. Here, $p_{1} \text{ and } p_{2}$ are the highest and the second-highest probability of detection for the output classes, respectively, when the random perturbations are repeatedly added to the input example. For example, if an input example is randomly perturbed 6 times and the predictions of the perturbed inputs are \{class-0, class-1, class-1, class-0, class-1, class-2\}, the corresponding values are $p_{1} = \frac{1}{2}$ and $p_{2} = \frac{1}{3}$. This is because class-1 is selected $\frac{1}{2}$ of the times (the class with the highest probability of being selected) and class-0 is selected $\frac{1}{3}$ of the times (the class with the second highest probability of being selected).

To analyze the impact of having a Trojan trigger on the value of $\delta=p_1-p_2$, we present the following theorem.

\begin{theorem}\label{th:delta-diff}
    Suppose we have a Trojan-infected classifier with a set of weight parameters $\theta$ which is perfectly trained to predict the ground truth values on benign examples while outputting the adversary's target class on any Trojan input. Assume also that the training data is drawn from the distribution $\mathcal{D}$ and each input example has $m$ replicas which are randomly perturbed. When $m = \infty$, the random perturbation sampled from $\mathcal{D}'$ makes the value of $\delta$ ($\delta = p_{1} - p_{2}$) for any Trojan example $\hat{x} + t$ larger than that for any benign example. Here, $\mathcal{D}'$ follows the same distribution as $\mathcal{D}$ with a different mean value set to $\mathbb{E}(\mathcal{D}) - \hat{x}$.
\end{theorem}

\begin{proof}
    Let's first focus on the training process of the Trojan-infected classifier. The training process can be represented by the following optimization problem: 
    \begin{align} 
        \theta = \underset{\theta}{\arg\min} (w_{1}\mathcal{L}(\hat{x}, y, \theta) + w_{2}\mathcal{L}(\hat{x}+t, y_{t}, \theta)) \label{eq:trojan-loss}
    \end{align}
    \guanxiong{Here, $w_{1} \text{ and } w_{2}$ are the weights of two loss terms.} Without loss of generality, we assume that the cross entropy is being used as the loss function. Therefore, the two loss terms could be written as:
    \begin{align}
        \mathcal{L}(\hat{x}, y, \theta) = \underset{\hat{x} \sim X}{\mathbb{E}} (- \log (f_{y}(\hat{x}))) \label{eq:benign-loss-term}\\
        \mathcal{L}(\hat{x}+t, y_{t}, \theta) = \underset{\hat{x} \sim X}{\mathbb{E}} (- \log (f_{y_{t}}(\hat{x}+t))) \label{eq:trojan-loss-term}
    \end{align}
    Here, Eq. \ref{eq:benign-loss-term} is used when the input is a benign example while Eq. \ref{eq:trojan-loss-term} is used for Trojan examples.
    
    Since each pixel's value among training examples, $X$, is drawn from the distribution $\mathcal{D}$, we can rewrite Eq. \ref{eq:trojan-loss-term} as follows:
    \begin{align}
        \mathcal{L}(\hat{x}+t, y_{t}, \theta) = \underset{\eta \sim \mathcal{D}}{\mathbb{E}} (- \log (f_{y_{t}}(t + \eta))) \label{eq:rewrite-loss-2}
    \end{align}
    Here, $\eta$ represents the random perturbation. Recall in Section \ref{sec:background}, $f_{k}(\cdot) \in [0, 1]$. 
    \guanxiong{Since the Trojan-infected classifier predicts the target class on any Trojan input, we will have $f_{y_{t}}(t + \eta) > f_{k}(t + \eta) ~~~ \forall k \in \{0,...,K\} \backslash y_{t}$. Therefore, we have $\underset{\eta \sim \mathcal{D}}{\mathbb{E}} [f_{y_{t}}(t + \eta)] > \underset{\eta \sim \mathcal{D}}{\mathbb{E}} [f_{k}(t + \eta)] ~~~ \forall k \in \{0,...,K\} \backslash y_{t}$.}
    This means that Trojan trigger $t$ with any perturbation $\eta$ sampled from $\mathcal{D}$ could fool the Trojan-infected classifier to output the target $y_{t}$.
    
    Now we move to the inference stage. If a Trojan example is received during the inference, the probability to predict it to class-$k$ under random perturbation could be represented as $\underset{\eta \sim \mathcal{D}'}{\mathbb{E}} [f_{k}(\hat{x} + t + \eta)]$. If the distribution $\mathcal{D}'$ is generated by subtracting the constant value $\hat{x}$ from the mean of $\mathcal{D}$ (denoted as $\mathcal{D}' = f(\mathcal{D}, \hat{x})$), the prediction probability to target class, $y_{t}$, could be rewritten as:
    \begin{align} 
        \underset{\eta \sim \mathcal{D}'}{\mathbb{E}} [f_{y_{t}}(\hat{x} + t + \eta)] = \underset{\eta \sim \mathcal{D}}{\mathbb{E}} [f_{y_{t}}(t + \eta)] \label{eq:inf-trojan-1}
    \end{align}
    \guanxiong{Therefore, from Eq. \ref{eq:inf-trojan-1} we have $\forall \eta \sim \mathcal{D}'$:}
    \begin{align}
        &   f_{y_{t}}(\hat{x} + t + \eta) > \underset{k \neq y_{t}}{\max} f_{k}(\hat{x} + t + \eta) &   \forall k \in \{0,...,K\} \backslash y_{t}
    \end{align}
    Based on the definition, we have $p_{1} = 1$ and $p_{2} = 0$ which results in $\delta = p_{1} - p_{2} = 1$.
    
    Lastly, we show that none of benign examples can achieve $\delta = 1$ in the inference through contradiction. Under the random perturbation from the same distribution, $\mathcal{D}'$, we assume that $\delta = p_{1} - p_{2} = 1$ holds for a benign examples $\hat{x}$ with ground truth $y$. \guanxiong{Therefore, we have $\forall \eta \sim \mathcal{D}'$:}
    \begin{align} 
        f_{y}(\hat{x} + \eta) > \underset{k \neq y}{\max} f_{k}(\hat{x} + \eta) = 0 \text{\ \ \ \ } \forall k \in \{0,...,K\} \backslash y \label{eq:inf-benign-1}
    \end{align}
    Recall that the distribution $\mathcal{D}'$ is generated by subtracting the constant value $\hat{x}$ from the mean of $\mathcal{D}$. Therefore, Eq. \ref{eq:inf-benign-1} can be rewritten as:
    \begin{align} 
        f_{y}(\eta) > \underset{k \neq y}{\max} f_{k}(\eta) = 0 \text{\ \ \ \ } \forall k \in \{0,...,K\} \backslash y \label{eq:inf-benign-2}
    \end{align}
    This means that any $\eta$ sampled from distribution $\mathcal{D}$ is predicted to class-$y$. Given that $\mathcal{D}$ denotes the distribution of pixel's value in training data, this means that the Trojan-infected classifier predicts any training data to class-$y$. 
    \guanxiong{Eq. \ref{eq:inf-benign-2} contradicts the fact that the classifier predicts the ground truth on benign examples.}
\end{proof}

When the conditions hold, the theorem above states that the value of $\delta = p_1-p_2$ for Trojan examples will be equal to 1 and larger than that for any benign example. Therefore, under the conditions presented in the theorem, i.e., perfect attacker, knowledge of the training data distribution, and $m=\infty$, we can decide that the input example is Trojan if $\delta=1$ and benign otherwise. Therefore, we can select the function we apply to $\delta$ to be $L=\delta$.

\guanxiong{However, the conditions in Theorem~\ref{th:delta-diff} are hard to be satisfied in reality because: (1) As a black-box defense, it is hard to know the data distribution $\mathcal{D}$. In our experiments, we found that Gaussian distribution is an efficient approximation of $\mathcal{D}$ as the distribution of pixel values often follows Gaussian distribution and can be normalized to a standard Gaussian distribution in convolutional neural network \cite{goodfellow2016deep}.}
(2) We can only run the algorithm with finite $m$. Since $p_{1}$ and $p_{2}$ follow Binomial distribution, \guanxiongVtwo{we can approximate the confidence interval for this results through using the Clopper-Pearson method introduced in \cite{brown2001interval}}. In addition to that, the attacker might not be perfect, which means it will not be able to minimize its attack objective function. Due to the above, we observe that the value of $\delta$ for some of the Trojan examples in Figure~\ref{fig:sigmoid-demo} (a) is below 1 (the green bars in the figure). It is worth noting that the plot in Fig.~\ref{fig:sigmoid-demo} (a) is generated with $\hat{L} = f(\delta) = \sigma \times (p_{1} - p_{2})$ where $\sigma$ is the standard deviation of the Gaussian noise. We include $\sigma$ since it is dynamically changing (detailed in later subsection), and this is the reason why the maximum value in Figure \ref{fig:sigmoid-demo} (a) is 0.2 rather than 1.
\guanxiong{(3) It is not guaranteed that the predictions on Trojan examples will always result in the target class. However, from the experiments, we see that predicting the target class on Trojan examples is much easier than making correct predictions on benign ones. For example, in Figure \ref{fig:prediction-heatmap}, we present the heatmaps of benign and Trojan examples. Each heatmap is a $10 \times 10$ matrix, where the rows represent the ground truth and columns represent the prediction results. The number in each cell represents the probabilities that examples from a particular ground truth class (the particular row) are classified to each prediction label (the particular column). We can see that in Figure \ref{sfig:testb} the numbers in the main diagonal are at most $0.9$ while most of the other cells are non-zero. On the other hand, in Figure \ref{sfig:testc} we only have $1.0$ in column 7. Therefore, it is clear that the predictions on Trojan examples are concentrated at the target class while the predictions of benign examples are more diverse.}

\begin{figure}[tb]
\begin{minipage}[c]{.48\linewidth}
    \begin{minipage}[c]{\textwidth}
    \centering
        \includegraphics[width=\linewidth]{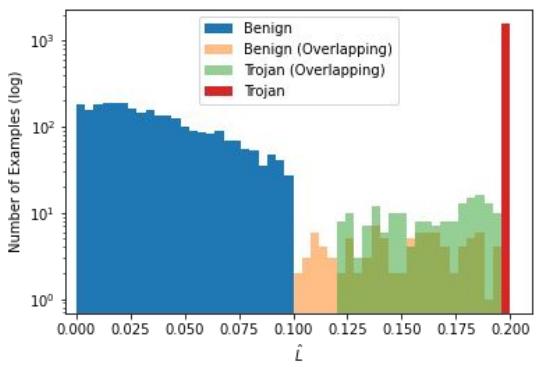}
    \end{minipage}
    \subcaption{Distribution of $\hat{L}$ (before applying sigmoid function)}
\end{minipage}
\hfill
\begin{minipage}[c]{.48\linewidth}
    \begin{minipage}[c]{\textwidth}
    \centering
        \includegraphics[width=\linewidth]{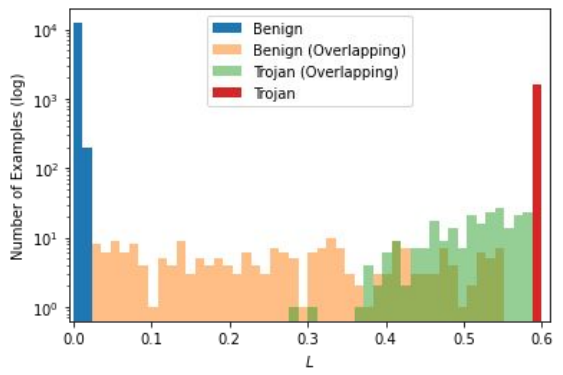}
    \end{minipage}
    \subcaption{Distribution of $L$ (after applying sigmoid function)}
\end{minipage}
\caption{The effect of applying the non-linear transform on the Prediction Confidence Bound\protect\footnotemark}
\label{fig:sigmoid-demo}
\end{figure}

\footnotetext{\guanxiongVtwo{The presented results are generated based on CIFAR-10 dataset under Trojan backdoor attack. The parameter setting and network are presented in Section IV.}}

\begin{figure}[tb]
\centering
\begin{minipage}[c]{.48\linewidth}
    \begin{minipage}[c]{\textwidth}
    \centering
        \includegraphics[width=\linewidth]{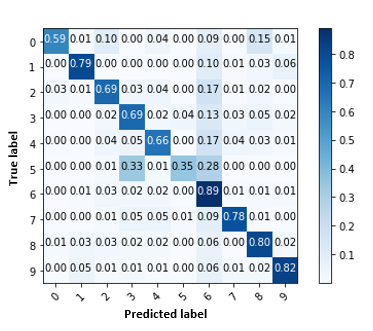}
    \end{minipage}
    \subcaption{Benign examples}
    \label{sfig:testb}
\end{minipage}
\hfill
\begin{minipage}[c]{.48\linewidth}
    \begin{minipage}[c]{\textwidth}
    \centering
        \includegraphics[width=\linewidth]{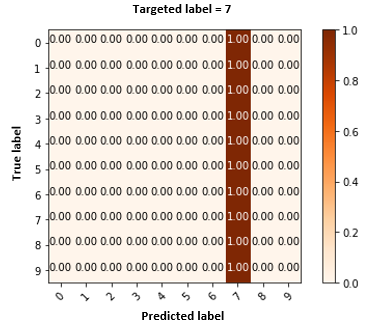}
    \end{minipage}
    \subcaption{Trojaned examples}
    \label{sfig:testc}
\end{minipage}
\caption{Heatmap of prediction on different examples}
\label{fig:prediction-heatmap}
\end{figure}

Even though the value of $\delta$ might not be equal to $1$ for Trojan examples when the conditions in Theorem~\ref{th:delta-diff} are not met, the main conclusion that the value of $\delta$ for any Trojan example is always larger than that of any benign example still generally holds. However, In Figure~\ref{fig:sigmoid-demo} (a) we can clearly see that the green bars that represent Trojan examples with $\delta < 1$ are very close to the orange bars representing benign examples. Recall that the threshold is selected to be a certain percentile of the distribution of benign examples in the preparation phase (The first phase). Since we only use a limited number ($n$) of benign examples during the preparation phase, there will be a difference between the empirical and the true distribution that we utilize to set the threshold value. In Figure \ref{fig:sigmoid-demo} (a), the overlapping of benign and Trojan examples are concentrated in a smaller range which makes the threshold very sensitive to the changes in the fitted distribution.

To mitigate this issue, we can apply a monotonic function to $\delta$ that can shift the distribution of the benign examples to the left-hand side of Figure~\ref{fig:sigmoid-demo} (a) and the distribution of Trojan examples to the right-hand side of the figure. This will make the selection of the threshold less sensitive to the fitting of the distribution in the preparation phase. To do that, we apply the sigmoid function on top of $\delta$ and derive the prediction confidence bound as follows.
\begin{align}
    L = \frac{1}{1 + e^{-d}} \text{\ \ \ where\ \ \ } d = \alpha \times [(p_{1} - p_{2}) \times \sigma - \beta] \label{eq:robust-bound}
\end{align}
Here, $\sigma$ represents the standard deviation of the random Gaussian noise while $\alpha$ and $\beta$ are the hyper-parameters. Through tunning the hyper-parameters ($\alpha$ and $\beta$) in Eq. \ref{eq:robust-bound}\footnote{\guanxiongVtwo{It is worth to note that the sigmoid function is tuned on benign examples only with the focus on reducing the residual error when the examples' values are fitted to a folded normal distribution.}}, we could align the center of the sigmoid function to the overlapping area. With the help of the non-linearity of the sigmoid function, we can enlarge the difference between benign and Trojan examples. It is clear in Figure \ref{fig:sigmoid-demo}(b) that the empirical distribution of the benign examples is pushed towards the lower end of $L$. Therefore, applying the sigmoid function results in the desired zoom-in effect to the overlapping area, as can be seen in Figure \ref{fig:sigmoid-demo} (b). It is worth mentioning that our method utilizes non-linear transformation enhances the performance of the proposed defense which is different from \cite{du2019robust} that designs the transformation as defense. In terms of defending Trojan backdoor, both \cite{du2019robust} and our method perform well on MNIST dataset. However, our method is successfully extended to larger datasets (e.g., CIFAR-10, GTSRB and CUB-200) which are not evaluated in \cite{du2019robust}. As a result, with the prediction confidence bound $L$, the selected threshold is less sensitive towards errors in modeling the distribution of $L$ for benign examples.

\begin{figure*}[tb]
\centering
\begin{minipage}[c]{\linewidth}
    \begin{minipage}[c]{\textwidth}
    \centering
        \includegraphics[width=\linewidth]{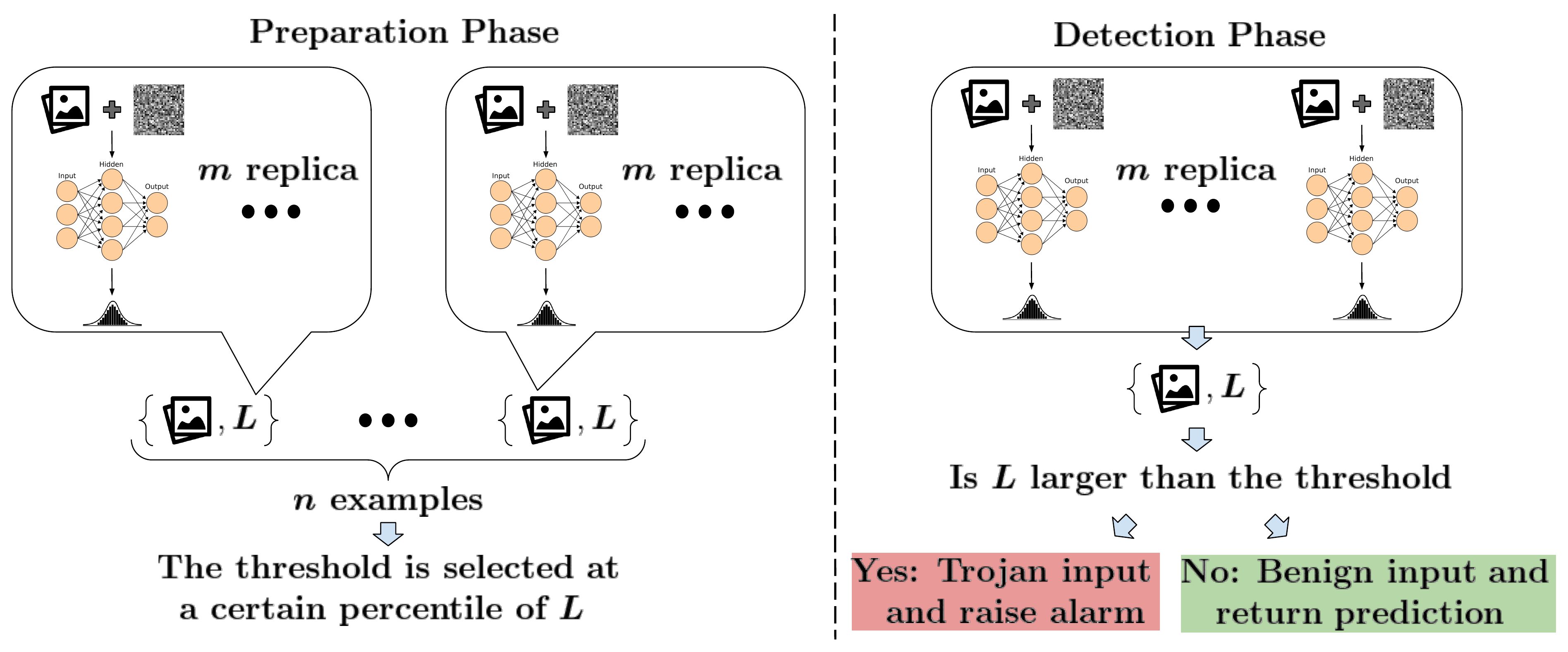}
    \end{minipage}
\caption{High-level view of the proposed defense}
\label{fig:defense-flow}
\end{minipage}
\vspace{-5mm}
\end{figure*}

\subsection{\textbf{\textsc{TrojDef}} Description}


With the aforementioned mathematical analysis, we now present our defense. As presented in Figure \ref{fig:defense-flow}, the proposed defense consists of two different phases. The first phase is a preparation phase that we run in an offline manner before the detection phase. During the first phase, we run \textbf{\textsc{TrojDef}} with a set of $n$ benign examples. Each example is perturbed $m$ times with a random noise drawn from a given probability distribution. Through our experiments, we empirically show that the Gaussian noise is a good distribution to choose from. Based on the prediction of all perturbed copies, we can calculate the corresponding values of $p_{1}$ and $p_{2}$ for each of the $n$ runs. Then, we further apply a function to the difference between $p_1$ and $p_2$ (i.e., $\delta = p_{1} - p_{2}$) in each of the $n$ runs. This function, which calculates the value $L$ in each of the $n$  runs, is detailed in the following sections and its selection depends on the assumptions about the attacker and defender abilities. After doing the above, we will have $n$ different $L$ values, and each is a result of applying the function to $\delta$ of each run. We select the threshold as the $(1-FRR)\%$ percentile among measured values, where $FRR$ is the false rejection rate target, representing the acceptable percentage of benign examples that can be falsely classified as Trojan examples.

The detection phase is performed in run time. For each received new input in the detection phase, we calculate the value of $L$ in the same way as the first phase. Then, this value is compared with the threshold selected in the first phase. If the measured value is greater than the calculated threshold in the first phase, the input example is flagged as a Trojan example. Otherwise, it is determined as a benign example. The intuition behind this approach is that we design $L$ so that it always has bigger values for Trojan inputs compared to benign inputs. Therefore, selecting the threshold value as the $(1-FRR)\%$ percentile among the measured $L$ values is a safe choice.

\subsection{\textbf{\textsc{TrojDef}} Algorithms}

The step-by-step process of the first phase of \textbf{\textsc{TrojDef}} is summarized in Algorithm \ref{algorithm:preparation-implementation}. \guanxiong{In the algorithm, the lines in blue represent the empirical enhancements that will be introduced in the next subsection.} \guanxiong{In lines 3-9}, we repeatedly perturb benign examples with random Gaussian noise. Then, \guanxiong{in lines 11-13}, the value of $L$ for each benign example is calculated. Finally, \guanxiong{in line 15}, the threshold value is selected to be higher than $(1-FRR)\times{100}\%$ of the values of $L$ for benign examples. 

The detailed process of the detection phase of \textbf{\textsc{TrojDef}} operating at the run time is detailed in the Algorithm \ref{algorithm:run-implementation}. \guanxiong{Similar to before, the empirical enhancements are in blue and will be detailed in the next subsection.} \guanxiong{In lines 1-8}, the input example is perturbed in the same way as what is done in phase 1 to calculate the corresponding $L$ value. Then, \guanxiong{in lines 9-16}, the calculated value is compared with the threshold selected in the previous phase. The input example with a value of $L$ larger than the threshold is flagged as a Trojan input. Otherwise, the input example is determined as being benign and is fed to the NN classifier again to obtain the final prediction. \guanxiongVtwo{Since generating Gaussian random noise is ignorable when compared with predicting the example, the total computation is $m$ times larger after applying the defense. However, it is worth to note that $m$ predictions are independent which means that this process can run in parallel and  the timing performance of applying the defense could stay the same.}

\subsection{\textbf{\textsc{TrojDef}} Implementation}
In this section, we provide the details of several practical enhancements to the basic algorithm above, especially when the input examples are images.

\subsubsection{Single Channel Perturbation}
From our empirical results, we notice that adding the random Gaussian noise blindly to the whole image may by far change the appearance of the Trojan trigger. Based on the conclusion drawn from \cite{li2020rethinking}, the changes in appearance or location of the trigger beyond a certain limit sharply decrease the attack success rate. To mitigate this issue, \textbf{\textsc{TrojDef}} takes an alternative way in that it perturbs only one channel with the Gaussian noise when the input is a multi-channel image (i.e., RGB image). For more details, we explicitly compare the performance of applying perturbation on blue channel and on all channels in Tables \ref{table:all-channel}, \ref{table:blue-vs-other}, and \ref{table:no-enhance} in the Appendix under the different settings and 1\% FRR threshold. It is clear that applying the perturbation on all channels performs poorly on some combinations of dataset, model, and trigger.

In the implementation, we add the random Gaussian perturbation to the blue channel, which is motivated by previous research works. It is demonstrated in~\cite{banu2016meta} that the blue channel in the RGB image is the darkest channel and contains a lower number of features compared with other channels. Moreover, the experiments in \cite{karahan2016image} show that the changes in prediction caused by modifying the blue channel are smaller than that caused by modifying other channels. Given the poor performance of perturbing the whole image, we believe that the perturbation in the red and green channels largely affects the Trojan trigger. Therefore, \textbf{\textsc{TrojDef}} only adds random Gaussian noise to the blue channel. Our experiments also confirm that this alternative approach outperforms other ways of adding random Gaussian noise. Table \ref{table:blue-vs-other} show the results of adding the same Gaussian noise on different single-channel under the same settings and 1\% FRR. It's clear that the performance sharply degenerates when the Gaussian noise is applied on red or green channel. This means that the adding Gaussian noise to red or green channel is an overkill since the Trojan trigger does not work either. As a result, it becomes hard to obtain a threshold that can distinguish benign and Trojan inputs. 


\begin{figure}
    \centering
    \includegraphics[width=0.45\textwidth]{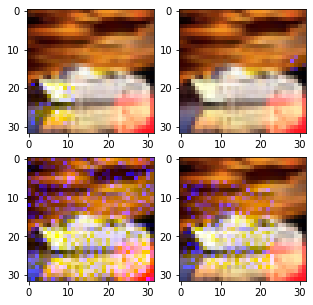}
    \caption{Perturbation with random location and size}
    \label{fig:noiseRand}
\end{figure}

\subsubsection{Randomizing the Location and Size of the Gaussian Perturbation}
As shown in Figure \ref{fig:noiseRand}, randomizing the location and size of the added random Gaussian perturbation is another trick that we apply to enhance the performance of \textbf{\textsc{TrojDef}}. Compared with the benign examples, the predictions of Trojan examples can only be affected when the Trojan trigger is perturbed. Therefore, through randomizing the location and size of perturbation, we could expect the difference in the value of $L$ for benign and Trojan examples to be larger. In the implementation, \textbf{\textsc{TrojDef}} randomly selects the location and size of the random Gaussian perturbation for each perturbed image. As shown in Figure \ref{fig:noiseRand}, we utilize a square area, and its size can be any integer value between 2 pixels to the size of the image. Depending on the size, the location is randomly selected starting from the top-left corner (i.e., [0,0]) to the limit that keeps the perturbation within the image area. 
\fatima{
In Table \ref{table:rand-size-exp}, we examined the effectiveness of randomizing the location and size of the Gaussian perturbation enhancements. It is clear that this enhancement highly affects the results because when applying the perturbation on the whole image it has a higher chance to change the appearance of the trigger.
}

\begin{algorithm} \caption{Preparation Phase of \textsc{\textbf{TrojDef}}} \label{algorithm:preparation-implementation}
\begin{algorithmic}[1]
\Require A trained classifier with weight parameter $\theta$ and an FRR
\Ensure Detection threshold $\tau$
\State Preparing $n$ different benign examples
\For{Each benign example $\hat{x}$}
    \State \textcolor{blue}{Flatten the pixel values in $\hat{x}$}
    \State \textcolor{blue}{Calculate the average of top-$k$ pixel values and store as $v$}
    \State \textcolor{blue}{Calculate $\sigma = -( S * \log_2 v )$}
    \For{$m$ iterations}
        \State \textcolor{blue}{Sample a random size perturbation $\eta$ from Gaussian distribution $\mathcal{N}(0, \sigma)$}
        \State \textcolor{blue}{Add $\eta$ to the blue channel of $\hat{x}$ at a random location}
        \State Store the prediction $C_{\theta}(\hat{x}+\eta)$
    \EndFor
    \State Calculate $p_{1}$ and $p_{2}$ for this example
    \State Calculate $d = \alpha \times [(p_{1} - p_{2}) \times \sigma - \beta]$
    \State Calculate and store prediction confidence bound $L = \frac{1}{1 + e^{-d}}$
\EndFor
\State Select the $\tau$ to be higher than the $(1-FRR)\times{100}\%$ percentile of the $L$ values.
\end{algorithmic}
\end{algorithm}
\begin{algorithm} \caption{Detection Phase of \textsc{\textbf{TrojDef}}} \label{algorithm:run-implementation}
\begin{algorithmic}[1]
\Require A trained classifier with weight parameter $\theta$, the threshold $\tau$, and an arbitrary input $x$
\Ensure The prediction
\State \textcolor{blue}{Flatten the pixel values in $x$}
\State \textcolor{blue}{Calculate the average of top-$k$ pixel values and store as $v$}
\State \textcolor{blue}{Calculate $\sigma = -( S * \log_2 v )$}
\For{$m$ iterations}
    \State \textcolor{blue}{Sample a random size perturbation $\eta$ from Gaussian distribution $\mathcal{N}(0, \sigma)$}
    \State \textcolor{blue}{Add $\eta$ to the blue channel of $x$ at a random location}
    \State Store the prediction $C_{\theta}(x+\eta)$
\EndFor
\State Calculate the $p_{1} \text{ and } p_{2}$ for $x$
\State Calculate $d = \alpha \times [(p_{1} - p_{2}) \times \sigma - \beta]$
\State Calculate the prediction confidence bound $L = \frac{1}{1 + e^{-d}}$
\If{$L > \tau$}
    \State Output the alarm that $x$ could be a Trojan input
\Else
    \State Output $C_{\theta}(x)$
\EndIf
\end{algorithmic}
\end{algorithm}

\subsubsection{Dynamic Standard Deviation}
Based on our experiments with a fixed value of $\sigma$ for the added Gaussian noise, we observe that the results are sensitive to the value of $\sigma$ in some cases \fatima{as illustrate in Table \ref{table:dyn-sigma}}. Depending on the combinations of the NN classifiers and Trojan triggers, using a fixed $\sigma$ value may work in some cases but fails in others since each case has different prediction confidence under the same perturbation. By making $\sigma$ dynamically changing based on the pixel values in each image, we are able to overcome this issue and achieve a good performance in separating the benign and Trojaned images. In our implementation, the following formula is used to calculate $\sigma$ for the added Gaussian noise to the pixels of each image:
\begin{align}
    \sigma = -( S * \log_2 v ) \label{eq:dynamic-sd}
\end{align}
Here, $S$ is a scalar, $v$ is the average of the largest $k$ pixel values in the whole image. To prevent $\sigma$ from getting a value outside of the $[0,1]$ range, we include default values to limit $\sigma$ to be within this range. \guanxiongVtwo{By utilizing Eq.\ref{eq:dynamic-sd}, the added noise could be controlled with respect to the visual content in the image. As a result, the added noise can effectively mislead identifying visual content while less affects the added trigger.} With this dynamic standard deviation, the values of $L$ for benign examples do not change much since the corresponding $\delta$ is small. For Trojan examples, \textbf{\textsc{TrojDef}} tends to use a smaller standard deviation when the pixel values are high (i.e., bright image). Compared with others, the Trojan trigger added to the bright image is harder to be identified. Therefore, applying noise with a smaller standard deviation helps Trojan examples to get a higher value of $\delta$ as well as $L$. It is worth noting that dynamically controlling the standard deviation values demonstrates the adaptability of \textbf{\textsc{TrojDef}} to better fit the input data, which is impossible with other state-of-the-art approaches, such as STRIP.

With all practical enhancements, the overall process from the preparation phase to making a prediction on input is summarized in Algorithms \ref{algorithm:preparation-implementation} and \ref{algorithm:run-implementation}. 
\fatima{
To show the enhancement of combining different empirical enhancements, we present evaluation results in Table \ref{table:no-enhance} that covers experiments with different combinations of presented empirical enhancements.
}

\begin{table*}[tb]
    \small
    \begin{center}
    \begin{tabular}{| c | c | c | c |  c | c | c | c | } 
      \hline
    dataset &  Convluation & Flatten & Dense &  Dropout & batch normalization & activation & Pooling \\
    \hline
       CIFAR-10  & 6 & 1 & 1 & 3 & √ & ReLU & 2 MaxPooling \\
     \hline
        GTSRB\tnote{*}  & 20 & 1 & 1 & 3 & √ & ReLU & 1 AveragePooling  \\
       \hline
     \end{tabular}
    \end{center}
    \caption{\textbf{\textsc{TrojDef}}-model architecture}
    \label{table:clf-archi}
\end{table*}

\section{Experimental Settings}\label{sec:setting}

In this section, we first introduce the datasets and the classifiers' architecture that are used. Then, we present the experiments and the calculated metrics.

\subsection{Datasets and Classifiers}

During the evaluation, we use the multiple benchmark datasets with different image size, number of samples and content to demonstrate that the advantage of our method over STRIP is independent from dataset:

$\bullet$\hspace{3mm}\textbf{MNIST: }Contains a total of 70K images and their labels. Each one is a $28 \times 28$ pixel, gray scale image of handwritten digits.


$\bullet$\hspace{3mm}\textbf{CIFAR-10: }Contains a total of 60K images and their labels. Each one is a $32 \times 32$ pixel, RGB image of animals or vehicles.

$\bullet$\hspace{3mm}\textbf{GTSRB: }Contains over 50K images and their labels. Each one is an RGB image of traffic signs with different sizes.

$\bullet$\hspace{3mm}\textbf{CUB-200: }Contains over 10K images with 200 classes. Each one is an RGB image of a bird with size of $300 \times 500$.

$\bullet$\hspace{3mm}\fatima{\textbf{ImageNet: }Contains over 14M images with 1000 classes. Each one is an RGB image.}

During the experiments, we include three different kinds of NN classifiers. (1) \textbf{STRIP-model:} The NN classifiers provided by the author of \cite{gao2019strip}. (2) \textbf{\textbf{\textsc{\textbf{\textsc{TrojDef}}}}-model:} The NN classifiers trained by us from scratch. (3) \textbf{3rd-party-model:} The ResNet-50 classifiers \cite{he2016deep} that are pre-trained by a 3rd party (we apply poisoned transfer learning to implant the Trojan backdoor). A brief summary of \textbf{\textsc{\textbf{\textsc{TrojDef}}}}-model architecture is presented in the Table \ref{table:clf-archi}.

\begin{figure*}[tb]
\centering
\begin{minipage}[c]{.13\linewidth}
    \begin{minipage}[c]{\textwidth}
    \centering
        \includegraphics[width=\linewidth]{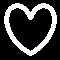}
    \end{minipage}
    \subcaption{heart}
    \label{sfig:testb}
\end{minipage}
\begin{minipage}[c]{.13\linewidth}
    \begin{minipage}[c]{\textwidth}
    \centering
        \includegraphics[width=\linewidth]{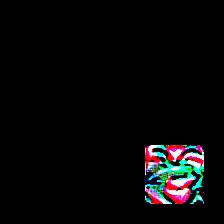}
    \end{minipage}
    \subcaption{face}
    \label{sfig:testc}
\end{minipage}
\begin{minipage}[c]{.13\linewidth}
    \begin{minipage}[c]{\textwidth}
    \centering
        \includegraphics[width=\linewidth]{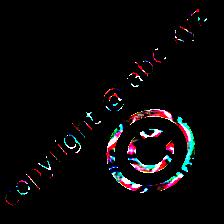}
    \end{minipage}
    \subcaption{watermark}
    \label{sfig:testd}
\end{minipage}
\begin{minipage}[c]{.13\linewidth}
    \begin{minipage}[c]{\textwidth}
    \centering
        \includegraphics[width=\linewidth]{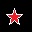}
    \end{minipage}
    \subcaption{star}
    \label{sfig:testd}
\end{minipage}
\begin{minipage}[c]{.13\linewidth}
    \begin{minipage}[c]{\textwidth}
    \centering
        \includegraphics[width=\linewidth]{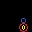}
    \end{minipage}
    \subcaption{bottle}
    \label{sfig:testd}
\end{minipage}
\begin{minipage}[c]{.13\linewidth}
    \begin{minipage}[c]{\textwidth}
    \centering
        \includegraphics[width=\linewidth,height=\linewidth]{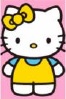}
    \end{minipage}
    \subcaption{Hello Kitty}
    \label{sfig:testd}
\end{minipage}
\begin{minipage}[c]{.13\linewidth}
    \begin{minipage}[c]{\textwidth}
    \centering
        \includegraphics[width=\linewidth,height=\linewidth]{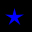}
    \end{minipage}
    \subcaption{blue star}
    \label{sfig:blueStar}
\end{minipage}
\caption{Trojan triggers used in the experiments}
\label{fig:trigger-visualization}
\end{figure*}

\subsection{Experiments and Metrics}

We compare \textbf{\textsc{\textbf{\textsc{TrojDef}}}} to STRIP due to the following reasons: (1) To the best of our knowledge, STRIP is the only black-box defense method, (2) STRIP achieves similar performance to other state-of-the-art white-box defenses as indicated in \cite{gao2019strip}. To comprehensively compare \textbf{\textsc{\textbf{\textsc{TrojDef}}}} with STRIP, we evaluate both defense methods on the three different models that are introduced before (i.e., STRIP-model, \textbf{\textsc{TrojDef}}-model, and 3rd-party-model). When evaluating with the STRIP-model, we try different training hyper-parameters. Moreover, the experiments with \textbf{\textsc{TrojDef}}-model and 3rd-party-model also include new Trojan triggers. Lastly, to explore the generalizability of \textbf{\textsc{\textbf{\textsc{TrojDef}}}} to different types of noise distributions, we run some of the experiments with Laplacian noise instead of Gaussian noise.

Throughout the experiments, we mainly focus on four different metrics. Among these metrics, we utilize the classification accuracy (\textbf{Acc}) and attack success rate (\textbf{Attack-Acc}) to evaluate the NN classifier that is infected by the Trojan attack.
\begin{itemize}
    \item \textbf{Acc: }The percentage of correctly classified benign examples over all benign examples.
    \item \textbf{Attack-Acc: }The percentage of Trojan examples that are classified into the adversary's target class when no defense is applied.
\end{itemize}
A Trojan infected NN classifier is trained to achieve high Acc and Attack-Acc simultaneously. The high Acc objective is to ensure that the classifier is of high quality to be adopted and used, while the high Attack-Acc objective ensures a successful attack. 

During the evaluation of the defense methods, we use the false acceptance rate (\textbf{FAR}) and the false rejection rate (\textbf{FRR}) as the performance metrics.
\begin{itemize}
    \item \textbf{FAR: }The percentage of Trojan examples that can pass the deployed defense method. The lower the FAR, the better the defense.
    \item \textbf{FRR: }The percentage of benign examples that are accidentally rejected by the deployed defense method. The lower the FRR, the better the defense.
\end{itemize}
Unless otherwise specified, we test both \textbf{\textsc{\textbf{\textsc{TrojDef}}}} and STRIP with a threshold value of the 99 percentile among benign examples. In other words, the FRR for both defenses is fixed at 1\%. Therefore, in the evaluation results, a better defense method should have a lower value of FAR.

Finally, we visualize the Trojan triggers used in the experiments in Figure \ref{fig:trigger-visualization}. When any of these triggers is mentioned, we use the caption of that trigger to refer to it.

\section{Experimental Results}\label{sec:results}

As we mentioned before, our experiments firstly evaluate the performance of \textbf{\textsc{TrojDef}} and STRIP on STRIP-model, \textbf{\textsc{\textbf{\textsc{TrojDef}}}}-model, and 3rd-party-model. Then, we further explore the performance of \textbf{\textsc{TrojDef}} under different settings which include (1) using smaller FRR rates, (2) adding noise that is drawn from a Laplacian random variable, and (3) defending a blue channel Trojan trigger. Lastly, we also compare the performance of our proposed black-box defense with the white-box approaches.

\subsection{Evaluation on STRIP-model}

The first part of the results is generated when STRIP-model is being used. These experiments strictly follow the original settings that are presented in \cite{gao2019strip}. The NN classifiers used in this subsection of experiments are provided directly by the authors of \cite{gao2019strip}. As STRIP has a very high detection accuracy on this model, through the experiments in this subsection, \guanxiong{we try to compare the proposed \textbf{\textsc{TrojDef}} with STRIP on the conventional experiments (i.e., the experiments conducted in STRIP work). The evaluation results are summarized Table \ref{table:conventional-res} \footnote{\fatima{We also present the FAR values under different selected FRR rates in Figure \ref{fig:strip_model_ROC} in the Appendix.}}}

\begin{table}[tb]
    \scriptsize
    \begin{center}
    \begin{tabular}{ c  c | c  c  c  c }
    \hline \hline
    \multirow{2}{*}{Dataset} & \multirow{2}{*}{Trigger} & \multirow{2}{*}{Acc} & \multirow{2}{*}{Attack-Acc} & \multicolumn{2}{c}{FAR} \\
    &  &  &  & STRIP & \textbf{\textsc{TrojDef}} \\
    \hline
    MNIST & "heart" & 99.02\% & 99.99\% & 0.1\% & 0\% \\
    \multirow{2}{*}{CIFAR-10} 
    & "face" & 83.84\% & 100\% & 0\% & 0\% \\
    & "watermark" & 82.35\% & 100\% & 0\% & 0\% \\
    %
    \hline \hline
    \end{tabular}
    \end{center}
    \caption{Results of the conventional experiments}
    \label{table:conventional-res}
\end{table}

Based on the value of Acc and Attack-Acc presented in Table \ref{table:conventional-res}, it is clear that the NN classifiers have been infected by the Trojan attack. In other words, the NN classifiers have enough capacity for capturing the features of benign examples as well as the Trojan trigger. These results validate that the performance of defense methods measured on top of the NN classifiers are reliable.

Under each combination of the dataset and Trojan trigger, we present the value of FAR for both \textbf{\textsc{TrojDef}} and STRIP. We can see that both defenses achieve 0\% FAR. Compared with the results presented in \cite{gao2019strip}, the performance of our reproduced STRIP is validated. More importantly, based on the conventional experiments, \textbf{\textsc{TrojDef}} achieves the same performance level as that of STRIP. In other words, there is no difference in terms of performance on conventional experiments between \textbf{\textsc{\textbf{\textsc{\textbf{\textsc{TrojDef}}}}}} and STRIP. However, in the following subsection, we can see that \textbf{\textsc{\textbf{\textsc{\textbf{\textsc{TrojDef}}}}}} outperforms STRIP when these experimental settings change.

\guanxiong{In addition to directly utilizing the STRIP-model, we also expand the experiments to evaluate the two defense methods when the hyper-parameters of the NN classifiers are changed. Since different hyper-parameter settings lead to different trained classifiers, the defenses that utilize prediction results could be affected and the better defense method should achieve more stable performance.} Here, we use the same architecture as STRIP-model but train it with different hyper-parameters. In these experiments, we choose three different hyper-parameters which include training epoch (\textbf{epoch}), learning rate (\textbf{lr}), and batch size (\textbf{bs}). We select the value of training epoch to be either 12 or 20. For learning rate, the possible values are $1e^{-4}$, $1e^{-3}$, and $3e^{-3}$. The batch size value varies between 60, 128, and 200. It is worth noting that these experiments are performed on MNIST dataset with "heart" trigger. The results of both \textbf{\textsc{TrojDef}} and STRIP are presented in Table \ref{table:expanded-hyper}.

\begin{table}[tb]
    \scriptsize
    \begin{center}
    \begin{tabular}{ c | c  c  c  c }
    \hline \hline
    \multirow{2}{*}{Hyper-parameters} & \multirow{2}{*}{Acc} & \multirow{2}{*}{Attack-Acc} & \multicolumn{2}{c}{FAR} \\
    &  &  & STRIP & \textbf{\textsc{TrojDef}} \\
    \hline
    epoch = 12 & 98.75\% & 99.54\% & 0.3\% & 0\% \\
    epoch = 20 & 98.96\% & 100\% & 17.05\% & 0\% \\
    \hline
    lr = $1e^{-4}$ & 98.76\% & 99.54\% & 20\% & 0\% \\
    lr = $1e^{-3}$ & 98.96\% & 100\% & 17.05\% & 0\% \\
    lr = $3e^{-3}$ & 98.66\% & 99.93\% & 1.05\% & 0\% \\
    \hline
    bs = 64 & 98.64\% & 100\% & 0.1\% & 0\% \\
    bs = 128 & 98.96\% & 100\% & 17.05\% & 0\% \\
    bs = 200 & 99.03\% & 100\% & 8.40\% & 0\% \\
    \hline \hline
    \end{tabular}
    \end{center}
    \caption{Performance of the defenses when the NN classifier is trained with different hyper-parameters}
    \label{table:expanded-hyper}
\end{table}

From the results, it is clear that \textbf{\textsc{TrojDef}} achieves more stable performance than that of STRIP when different hyper-parameters are used. Moreover, throughout the experimental results, \textbf{\textsc{TrojDef}} always achieves lower FAR value than that of STRIP. In addition, the FAR value of STRIP has a much obvious fluctuation compared to that of \textbf{\textsc{TrojDef}}. For example, the FAR for STRIP changes from 0.10\% to 17.05\% when the batch size changes from 60 to 128. When the learning rate changes, the FAR values for STRIP reach as high as 20\%. Basically, when different hyper-parameter settings are applied, the model with the same architecture may converge to different weight parameters. The results in Table \ref{table:expanded-hyper} show that only the changes in weight parameters are enough to largely degenerate the performance of STRIP. It is worth noting that the owner of the model is the one who decides the hyper-parameter settings, and there are always more than one setting that could work. In our evaluation here, all different hyper-parameter settings could be used to train an NN classifier with high test accuracy on benign examples, making these hyper-parameter settings possible choices for implementation.

\subsection{Evaluation on \textbf{\textsc{TrojDef}}-model}

In this part of the experiments, we evaluate both defenses (\textbf{\textsc{TrojDef}} and STRIP) in a broader range of settings. More specifically, we utilize (1) the \textbf{\textsc{TrojDef}}-model which has a different architecture than the model in the previous subsection, (2) the GTSRB dataset which is not evaluated in \cite{gao2019strip}, and (3) new Trojan triggers (i.e., "bottle" and "star"). The evaluation results are summarized in Table \ref{table:sensitivity-res-1} \footnote{\fatima{We also present the FAR values under different selected FRR rates in Figure \ref{fig:our_model_ROC} in the Appendix.}}.

\begin{table}[tb]
    \scriptsize
    \begin{center}
    \begin{tabular}{ c  c | c  c  c  c }
    \hline \hline
    \multirow{2}{*}{Dataset} & \multirow{2}{*}{Trigger} & \multirow{2}{*}{Acc} & \multirow{2}{*}{Attack-Acc} & \multicolumn{2}{c}{FAR} \\
    &  &  &  & STRIP & \textbf{\textsc{TrojDef}} \\
    \hline
    \multirow{4}{*}{CIFAR-10} 
    & "face" & 85.73\% & 100\% & 0\% & 0\% \\
    & "watermark" & 85.61\% & 100\% & 0\% & 0\% \\
    & "bottle" & 84.82\% & 99.30\% & 1.10\% & 0.15\% \\
    & "star" & 84.76\% & 100\% & 0\% & 0\% \\
    \hline
    \multirow{4}{*}{GTSRB} 
    & "face" & 99.85\% & 100\% & 100\% & 0\% \\
    & "watermark" & 99.80\% & 100\% & 100\% & 0\% \\
    & "bottle" & 99.90\% & 100\% & 100\% & 0.05\% \\
    & "star" & 99.89\% & 100\% & 100\% & 0\% \\
    \hline \hline
    \end{tabular}
    \end{center}
    \caption{Evaluation results of the defenses on \textbf{\textsc{TrojDef}}-model}
    \label{table:sensitivity-res-1}
\end{table}

From the values of Acc and Attack-Acc, it is clear that the Trojan backdoor has been successfully implanted to \textbf{\textsc{TrojDef}}-model. Also, from the FAR values in Table \ref{table:sensitivity-res-1}, we see the following.
\begin{enumerate}
    \item When changing from the STRIP-model to \textbf{\textsc{TrojDef}}-model, some of the FAR values of STRIP increase from 0\% to 100\% even for those triggers used in \cite{gao2019strip}.
    \item Compared with STRIP, \textbf{\textsc{TrojDef}} achieves more stable performance. The value of FAR does not change more than 0.15\% regardless of the changes in the classifiers or the Trojan triggers.
\end{enumerate}

The evaluation results in Table \ref{table:sensitivity-res-1} demonstrate clear issues regarding the performance of STRIP. When the NN classifier changes, the performance of STRIP may suffer a significant degeneration. We believe the following reason is related to this issue. \guanxiong{When the architecture is changed, classifiers trained on the same poisoned dataset are different.} Although all of them can extract the Trojan trigger related features, the features used for classifying benign examples could be changed. 
\guanxiong{As a result, some of these classifiers become more sensitive towards the perturbation. In other words, when using the same hold-out data (i.e., benign examples prepared for superimposition process) on such classifiers, the entropy values for benign and Trojan examples are indistinguishable.}

Although fine-tuning could be a solution to this issue, the design of STRIP makes it very difficult if not impossible to perform fine-tuning. Recall that to fine-tune STRIP, we need to collect new hold-out dataset \cite{gao2019strip}. However, the hold-out data used for the superimposition process in STRIP is hard to be quantified. In other words, when collecting new hold-out data, there is no clear guidance about what the new hold-out data should be. Therefore, we think that fine-tuning STRIP is very difficult if not impossible and the issue of unstable performance is unavoidable.


%
\begin{table}[tb]
    \scriptsize
    \begin{center}
    \begin{tabular}{ c  c | c  c  c  c }
    \hline \hline
    \multirow{2}{*}{Dataset} & \multirow{2}{*}{Trigger} & \multirow{2}{*}{Acc} & \multirow{2}{*}{Attack-Acc} & \multicolumn{2}{c}{FAR} \\
    &  &  &  & STRIP & \textbf{\textsc{TrojDef}} \\
    \hline
    \multirow{4}{*}{CIFAR-10} 
    & "face" & 93.12\% & 99.34\% & 100\% & 0\% \\
    & "watermark" & 93.56\% & 99.90\% & 0\% & 0\% \\
    & "bottle" & 93.82\% & 89.48\% & 24.50\% & 19.5\% \\
    & "star" & 93.54\% & 99.62\% & 0\% & 0\% \\
    \hline
    \multirow{4}{*}{GTSRB} 
    & "face" & 98.03\% & 99.30\% & 100\% & 0\% \\
    & "watermark" & 97.46\% & 99.95\% & 100\% & 0\% \\
    & "bottle" & 98.26\% & 99.84\% & 0\% & 0.05\% \\
    & "star" & 98.96\% & 98.04\% & 100\% & 0\% \\
    \hline
    \multirow{2}{*}{CUB-200} & "face" & 61.74\% & 99.14\% & 100\% & 1.15\% \\
    & "watermark" & 62.63\% & 99.86\% & 3.59\% & 0\% \\
     \hline
    \multirow{2}{*}{ImageNet} & "face" & 60.28\% & 27.67\% & 80.15\% & 51.5\% \\
    & "watermark" &  60.40\% & 42.75\% & 80.95\% & 45.35\% \\
    \hline \hline
    \end{tabular}
    \end{center}
    \caption{Evaluation results of the defenses on the 3rd-party-model}
    \label{table:sensitivity-res-2}
\end{table}

\subsection{Evaluation on 3rd-party-model}
In the third part of the experiments, we evaluate \textbf{\textsc{TrojDef}} and STRIP on the 3rd-party-model. The 3rd-party-model brings new angle to the evaluation of the two defenses because of the following:
\begin{itemize}
    \item Compared with the \textbf{\textsc{\textbf{\textsc{TrojDef}}}}-model, the 3rd-party-model is trained in a different way. These NN classifiers are pre-trained on ImageNet data. As a result, the NN classifiers are likely to extract different and more general features than those trained with only the target dataset (e.g. CIFAR-10 and GTSRB).
    
    \item With the development of model sharing platforms (e.g. GitHub and ``Paper with Code''), model reusing is becoming a popular choice especially when a large scale NN classifier is needed. Therefore, the evaluation with a specific focus on a 3rd-party-model is an interesting and important topic.
\end{itemize}

To closely reflect the real-world scenarios, the 3rd-party-model utilizes the ResNet50 NN classifier and is pre-trained on ImageNet data until it converges. After that, we apply transfer learning with these NN classifiers and the poisoned dataset. 
\guanxiong{It is also worth mentioning that our evaluation includes the CUB-200 dataset. This dataset contains images with pixel size around $300 \times 500$ which is the same level as the VGG-Face\cite{huang2008labeled} and ImageNet \cite{imagenet_cvpr09}. Therefore, the evaluation results on CUB-200 dataset also show the generalizability of \textbf{\textsc{TrojDef}}.} 
\fatima{
Last but not the least, we also conduct evaluation with ImageNet dataset to further demonstrate the effectiveness of \textbf{\textsc{TrojDef}}. 
By comparing the evaluation results in Table \ref{table:sensitivity-res-2} \footnote{\fatima{We also present the FAR values under different selected FRR rates in Figure \ref{fig:3rd_model_ROC} in the Appendix.}}, the significant advantage of \textbf{\textsc{TrojDef}} over STRIP still holds. In 9 out of 12 experiments, \textbf{\textsc{TrojDef}} outperforms STRIP (i.e. achieves much lower FAR values), while in other two experiments, both approaches achieve exactly 0\% FAR value. Also, in the experiment with GTSRB dataset and "bottle" trigger, both \textbf{\textsc{TrojDef}} and STRIP can achieve nearly 0\% FAR.
}

\fatima{
It is worth noting that the Attack-Acc on ImageNet is much lower than other datasets.
The reason is that 3rd-party-model is fully trained on ImageNet dataset without attack and we only retrain it a limited number of epochs with backdoor examples.
However, we still observe a large advantage of using \textbf{\textsc{TrojDef}} compared with STRIP in terms of FAR.
}

The 3rd-party-model is more challenging. Although \textbf{\textsc{TrojDef}} still outperforms STRIP, \guanxiongVtwo{it can only achieve about 20\% FAR in one out of 10 experiments}, while achieving very close to perfect accuracy (0\% FAR) on the remaining 9 experiments. STRIP on the other hand performs poorly on this dataset. In other words, the performance of \textbf{\textsc{TrojDef}} degenerates on one of the cases of the 3rd-party-model. We believe the following two reasons explain this observation.
\begin{enumerate}
    \item The Trojan backdoor is implanted to the 3rd-party-model through transfer learning which barely modifies the extracted features. Therefore, the 3rd-party-model learns the Trojan trigger by a set of existing features which is not as stable as other models that identify the Trojan trigger as a fundamental feature \cite{liu2018fine}. As a validation, we can see that the Attack-Acc value on 3rd-party-model is slightly lower than that for other models.
    
    \item The NN classifiers used in 3rd-party-model are pre-trained on a large-scale dataset (e.g., ImageNet) until convergence. To achieve solid performance, these pre-trained NN classifiers are usually optimized to perform consistently even under a certain level of perturbation. As a result, the predictions of some benign examples are quite confident and the added noise level might not be enough to fool the classifier with benign inputs. 
\end{enumerate}

Combining these reasons, we could expect the value of $L$ on benign examples to become larger while the value of $L$ for Trojan examples to become smaller \guanxiong{when the 3rd-party-model is being used}. As a result, it is clear that the overlapping between benign and Trojan examples becomes serious \guanxiong{in this evaluation}. It is worth to note that the aforementioned challenge is not only for \textbf{\textsc{TrojDef}} but also a threat to other defenses that depend on prediction confidence. Therefore, we believe that using the 3rd-party-model is a challenging and important evaluation given the defense methods (i.e., STRIP and \textbf{\textsc{TrojDef}}). Nonetheless, \textbf{\textsc{\textbf{\textsc{\textbf{\textsc{TrojDef}}}}}} achieves decent performance on this model.

\begin{table}[tb] 
    \footnotesize
    \begin{center}
    \begin{tabular}{ c  c |  c  c  c }
    \hline \hline
    \multirow{2}{*}{Dataset} & \multirow{2}{*}{Trigger} & 
    \multirow{2}{*}{FRR} & \multicolumn{2}{c}{FAR} \\
    &  &   & STRIP & \textsc{\textbf{\textbf{\textsc{\textbf{\textsc{TrojDef}}}}}} \\
    \hline
    \multirow{16}{*}{CIFAR-10} 
    &  \multirow{4}{*}{"face"}  &   0.25\% & 0\% & 0\% \\
     &  &   0.5\% & 0\% & 0\% \\
       &   & 0.75\% & 0\% & 0\% \\
    & &   1\% & 0\% & 0\% \\
    \cline{3-5}
    & \multirow{4}{*}{"watermark"} &  0.25\% & 100\% & 0\% \\
    & & 0.5\% & 0\% & 0\% \\
    & & 0.75\% & 0\% & 0\% \\
    & & 1\% & 0\% & 0\% \\
    \cline{3-5}
    & \multirow{4}{*}{"bottle"} & 0.25\% & 100\% & 0.15\% \\
    &  &  0.5\% & 100\% & 0.15\% \\
    &  &  0.75\% & 100\% & 0.15\% \\
    &  &   1\% & 1.10\% & 0.15\% \\
    \cline{3-5}
    & \multirow{4}{*}{"star"} & 0.25\% & 100\% & 0\% \\
    &   & 0.5\% & 100\% & 0\% \\
     &   & 0.75\% & 100\% & 0\% \\
      &   & 1\% & 0\% & 0\% \\
    \hline
    \multirow{16}{*}{GTSRB} 
    & \multirow{4}{*}{"face"} &  0.25\% & 100\% & 0\% \\
     &   & 0.5\% & 100\% & 0\% \\
      &  & 0.75\% & 100\% & 0\% \\
       &  & 1\% & 100\% & 0\% \\
    \cline{3-5}
    & \multirow{4}{*}{"watermark"} & 0.25\%& 100\% & 0\% \\
       & &  0.5\% & 100\% & 0\% \\
      & &   0.75\% & 100\% & 0\% \\
       & &   1\% & 100\% & 0\% \\
    \cline{3-5}
    & \multirow{4}{*}{"bottle"} &  0.25\% & 100\% & 0.05\% \\
   &   & 0.5\% & 100\% & 0.05\% \\
   &   & 0.75\% & 100\% & 0.05\% \\
   & &  1\% & 100\% & 0.05\% \\
    \cline{3-5}
    & \multirow{4}{*}{"star"} &  0.25\% & 100\% & 0\% \\
       &  & 0.5\% & 100\% & 0\% \\
      &   & 0.75\% & 100\% & 0\% \\
       & &  1\% & 100\% & 0\% \\
    \hline \hline
    \end{tabular}
    \end{center}
    \caption{Evaluation results of the defenses on \textbf{\textsc{TrojDef}}-model under different FRR values}
    \label{table:diff-frr-1}
\end{table}

\begin{table}[tb] 
    \footnotesize
    \begin{center}
    \begin{tabular}{ c  c |  c  c  c }
    \hline \hline
    \multirow{2}{*}{Dataset} & \multirow{2}{*}{Trigger} &  \multirow{2}{*}{FRR} & \multicolumn{2}{c}{FAR} \\
    &   &  & STRIP & \textbf{\textsc{TrojDef}} \\
    \hline
    \multirow{19}{*}{CIFAR-10} 
    & \multirow{4}{*}{"face"}  & 0.25\% & 100\% & 0\% \\
     & & 0.5\% & 100\% & 0\% \\
     & & 0.75\% & 100\% & 0\% \\
     & & 1\% & 100\% & 0\% \\
    \cline{3-5}
    & \multirow{4}{*}{"watermark"} &  0.25\% & 0\% & 0\% \\
     &  &  0.5\% & 0\% & 0\% \\
      &  &  0.75\% & 0\% & 0\% \\
       &  &  1\% & 0\% & 0\% \\
    \cline{3-5}
    & \multirow{4}{*}{"bottle"} & 0.25\% & 39.8\% & 32.55\% \\
     &  & 0.5\% & 28.249\% & 22.0\% \\
      &  & 0.75\% & 25.83\% & 22.0\% \\
       &  & 1\% & 24.50\% & 19.5\% \\
    \cline{3-5}
    & \multirow{4}{*}{"star"} &  0.25\% & 100\% & 0\% \\
    & &  0.5\% & 100\% & 0\% \\
     & &  0.75\% & 100\% & 0\% \\
      & &  1\% & 0\% & 0\% \\
    \hline
    \multirow{19}{*}{GTSRB} 
    & \multirow{4}{*}{"face"} & 0.25\% & 100\% & 0\% \\
     &  & 0.5\% & 100\% & 0\% \\
     &  & 0.75\% & 100\% & 0\% \\
     &  & 1\% & 100\% & 0\% \\
    \cline{3-5}
    & \multirow{4}{*}{"watermark"} & 0.25\% & 100\% & 0\% \\
     &  & 0.5\% & 100\% & 0\% \\
      &  & 0.75\% & 100\% & 0\% \\
       & & 1\% & 100\% & 0\% \\
    \cline{3-5}
    & \multirow{4}{*}{"bottle"} & 0.25\% & 100\% & 0.05\% \\
    & & 0.5\% & 100\% & 0.05\% \\
     & & 0.75\% & 0.05\% & 0.05\% \\
      & & 1\% & 0\% & 0.05\% \\
    \cline{3-5}
    & \multirow{4}{*}{"star"} &  0.25\% & 100\% & 0\% \\
     &  &  0.5\% & 100\% & 0\% \\
      &  &  0.75\% & 100\% & 0\% \\
       &  &  1\% & 100\% & 0\% \\
    \hline
    \multirow{9}{*}{CUB-200} & \multirow{4}{*}{"face"} & 0.25\% & 100\% & 1.5\% \\
      &  & 0.5\% & 100\% & 1.25\% \\
       &  & 0.75\% & 100\% & 1.25\% \\
        &  & 1\% & 100\% & 1.15\% \\
    \cline{3-5}
     & \multirow{4}{*}{"watermark"} & 0.25\% & 4.9\% & 0.05\% \\
     & & 0.5\% & 4.1\% & 0\% \\
      & & 0.75\% & 3.8\% & 0\% \\
         & & 1\% & 3.59\% & 0\% \\
    \hline \hline
    \end{tabular}
    \end{center}
    \caption{Evaluation results of the defenses on the 3rd-party-model under different FRR}
    \label{table:diff-frr-2}
\end{table}

\subsection{Using Different FRR Values}

\guanxiong{In previous experiments, we select the FRR value to be 1\%. However, in real world scenarios, the requirements on selected threshold varies and it is important to report the performance of the defense methods under different FRR values. Therefore, in this subsection, we repeat some of the experiments on both \textbf{\textsc{TrojDef}}-model and 3rd-party-model. Instead of using a fixed FRR value, we change it to be from the following set: $\{0.25,~0.5,~0.75,~1.0\}$. The results of these experiments are summarized in Tables \ref{table:diff-frr-1} and \ref{table:diff-frr-2}.}

\guanxiong{Based on the results it is clear that the FAR increases when the FRR decreases since there is a trade-off between detecting all potential Trojan inputs and reducing the false positive alarm. However, when we compare the detailed FAR values of STRIP and \textbf{\textsc{TrojDef}}, we can see that the \textbf{\textsc{TrojDef}} significantly outperforms STRIP. For example, on CIFAR-10 dataset with "star" trigger and \textbf{\textsc{TrojDef}}-model (Table \ref{table:diff-frr-1}), the proposed defense consistantly achieves $0.15\%$ FAR while the FAR of STRIP goes to $100\%$ when the $FRR$ is set to $0.75$ or lower. Similar observation can be obtained from Table \ref{table:diff-frr-2} as well (e.g., CIFAR-10 dataset with "bottle" trigger and 3rd-party-model). Compared with STRIP, these results show that \textbf{\textsc{TrojDef}} is a better defense method which can achieve very small FAR values when small target values are selected for FRR.}

\subsection{Using Laplacian Perturbation}

As presented in Section \ref{sec:defense}, \textbf{\textsc{TrojDef}} is designed to work with perturbations sampled from an arbitrary distribution as long as it closely approximates the distribution of the pixel values in training dataset. In order to validate this claim, we repeat the experiments with 3rd-party-model on CIFAR-10 and GTSRB datasets. During the evaluation, we replace the Gaussian perturbations with Laplacian ones. The results are summarized in Table \ref{table:laplacian-res}.

\begin{table}[tb]
    \footnotesize
    \begin{center}
    \begin{tabular}{ c | c  c  c  c }
    \hline \hline
    \multirow{2}{*}{Dataset} & \multicolumn{4}{c}{Trigger} \\
    & "face" & "watermark" & "bottle" & "star" \\
    \hline
    CIFAR-10 & 0\% & 0\% & 34.30\% & 0\% \\
    GTSRB & 0\% & 0\% & 0.05\% & 0\% \\
    \hline \hline
    \end{tabular}
    \end{center}
    \caption{Evaluation results of \textbf{\textsc{TrojDef}} on the 3rd-party-model and Laplacian Perturbation}
    \label{table:laplacian-res}
\end{table}

From these results, we can see that in 7 out of 8 cases using Laplacian perturbation \textbf{\textsc{TrojDef}} achieves the same FAR value as before. Only in the case of CIFAR-10 dataset and "bottle" trigger, using Laplacian perturbation degenerates the performance of \textbf{\textsc{\textbf{\textsc{\textbf{\textsc{TrojDef}}}}}}. We think that Gaussian perturbation is better than Laplacian perturbation for CIFAR-10 dataset. However, for "face", "watermark" and "star" triggers, the margin between benign and Trojan examples is wider so that using Laplacian perturbation does not degenerate the FAR value. While for "bottle" trigger, differentiating benign and Trojan examples is much harder and replacing the Gaussian perturbation with Laplacian perturbation leads to a lower FAR value. This can be validated by the results in Table \ref{table:sensitivity-res-2}. When using Gaussian perturbation, the FAR value is 22.10\% for "bottle" trigger while it is 0\% for the other triggers.

\subsection{Defending Blue Channel Trigger}

\begin{figure}
    \centering
    \includegraphics[width=0.45\textwidth]{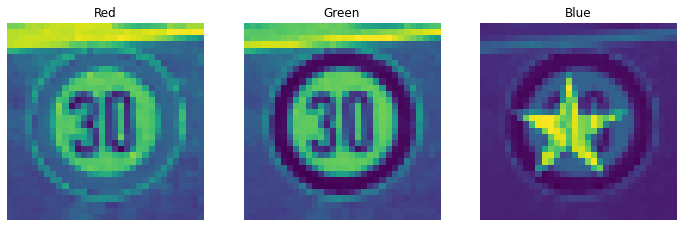}
    \caption{By channel view of blue channel trigger}
    \label{fig:RGBChan}
\end{figure}

Recall Sec \ref{sec:defense}-D, we present the single channel perturbation as one of the practical enhancements of our proposed defense. To complete our justification of adding perturbation to blue channel, in this subsection, we conduct an additional experiment to evaluate the performance of our proposed defense when the Trojan trigger lives in the blue channel. As shown in Figure \ref{fig:RGBChan}, we customized a "blue star" trigger which is added to only the blue channel of input examples. With this Trojan trigger, we evaluate the performance of \textbf{\textsc{TrojDef}} on different models as well as datasets. From the results summarized in Table \ref{table:blue-blue}, it is clear that the performance of \textbf{\textsc{TrojDef}} is not affected even if the Trojan trigger lives only in the blue channel.

\begin{table}[tb]
    \parbox{\linewidth}{
        \centering
        \begin{tabular}{ c  c | c  c   }
            \hline \hline
          Trigger&
           Model & Dataset &
           FAR \\
             \hline
            \multirow{6}{*}{"blue star"} &   \multirow{2}{*}{\textbf{\textsc{TrojDef}}}& CIFAR-10  
             & 0.0\% \\
             \cline{3-4}
            &   & GTSRB& 0.0\% \\
               \cline{2-4}
             &   \multirow{2}{*}{STRIP}& CIFAR-10  
             & 0.0\% \\
             \cline{3-4}
            &   & GTSRB& 0.0\% \\
             \cline{2-4}
             &   \multirow{2}{*}{3rd-party}& CIFAR-10  
             & 0.0\% \\
             \cline{3-4}
            &   & GTSRB& 0.0\% \\
            \hline \hline

        \end{tabular}
        
        \caption{Performance of defending the blue channel trigger}
        \label{table:blue-blue}
}
\end{table}

\subsection{Compared with White-Box Defense} 

\fatima{In this experiment, we use the proposed defense in \cite{jin2020unified} and we refer to it as Mutation defense. Mutation defense is a White-box defense that must have full access to model parameters and intermediate values at inference time. It generates m mutated model by adding Gaussian noise to the weights of the fully-connected layers. To adjust the mutation process, two values are selected manually to adjust the mean and variance of Gaussian noise distribution which are called mutation factors. For each layer, the mean value of the Gaussian noise distribution is calculated by multiplying the mean mutation factor by the mean of the fully-connected layer weights and the variance value of the Gaussian noise distribution is calculated by multiplying the variance mutation factor by the maximum weight value in a fully-connected layer. The intuition behind this approach is that the Trojaned inputs appear to have higher sensitivity to mutations on a NN model than benign inputs. Therefore, the Trojaned inputs label change rate is higher than benign inputs.}

\fatima{We compare \textbf{\textsc{TrojDef}}  with Mutation defense in Table \ref{table:white-box-rs}. 
It is clear that the performance of Mutation defense fluctuates significantly when facing different combinations of dataset, model and trigger. 
Although we tune the mutation factors to mitigate this issue, our attempts fail especially on the CIFAR-10 dataset. 
Moreover, on GTSRB dataset with \textbf{\textsc{TrojDef}} model, the FAR of Mutation defense varies from $7.65\%$ to $28.80\%$ which confirms the unstable performance of this defense.
In general, from the results, we conclude that Mutation defense works in some of our evaluation cases while fails in other cases. Also, we found that tuning mutation factors is not enough to enhance Mutation defense in the poorly performed cases.}

\begin{table}[tb] 
     \footnotesize
    \begin{center}
    \begin{tabular}{ c  c |  c  c  c }
    \hline \hline
    \multirow{2}{*}{Dataset} & \multirow{2}{*}{Trigger} & 
    \multirow{2}{*}{model} & 
    \multicolumn{2}{c}{FAR} \\
    &  &   & Mutation & \textsc{\textbf{\textbf{\textsc{\textbf{\textsc{TrojDef}}}}}} \\
    \hline
   \multirow{2}{*}{ MNIST}
    &  "square"  & \cite{jin2020unified} model &  0.01\% & 0\%  \\
    &  "heart"  & \textsc{\textbf{\textbf{\textsc{\textbf{\textsc{TrojDef}}}}}} &   65.0\% & 0\%  \\
   
    \hline
    
    \multirow{7}{*}{CIFAR-10} 
    
    &  \multirow{3}{*}{"face"}  &STRIP &  100.0\% & 0\%  \\
      &  & \textsc{\textbf{\textbf{\textsc{\textbf{\textsc{TrojDef}}}}}}&   100.0\% & 0\% \\
       &   &3red-party& 100.0\% & 0\%  \\
   
    \cline{3-5}
    & \multirow{2}{*}{"watermark"} &  STRIP & 99.95\% & 0\% \\
    & &  \textsc{\textbf{\textbf{\textsc{\textbf{\textsc{TrojDef}}}}}} & 84.75\% & 0\% \\
    
    \cline{3-5}
    &"bottle" &\textsc{\textbf{\textbf{\textsc{\textbf{\textsc{TrojDef}}}}}} & 100.0\% & 0.15\% \\
    
    \cline{3-5}
    & "star" & \textsc{\textbf{\textbf{\textsc{\textbf{\textsc{TrojDef}}}}}} & 100.0\% & 0.0\% \\
   
    \hline
    \multirow{5}{*}{GTSRB} 
    & \multirow{2}{*}{"face"} & STRIP & 100.0\% & 0.0\% \\
     &   & \textsc{\textbf{\textbf{\textsc{\textbf{\textsc{TrojDef}}}}}} & 7.65\% & 0.0\% \\
      
    \cline{3-5}
    & "watermark" &\textsc{\textbf{\textbf{\textsc{\textbf{\textsc{TrojDef}}}}}} & 20.95\% & 0.0\% \\
       
    \cline{3-5}
    & "bottle" &  \textsc{\textbf{\textbf{\textsc{\textbf{\textsc{TrojDef}}}}}} & 12.95\% & 0.05\% \\
  
    \cline{3-5}
    & "star" &  \textsc{\textbf{\textbf{\textsc{\textbf{\textsc{TrojDef}}}}}} & 28.80\% & 0.0\% \\
   
    \hline \hline
    \end{tabular}
    \end{center}
    \caption{Evaluation results of the Mutation and \textbf{\textsc{TrojDef}} model}
    \label{table:white-box-rs}
\end{table}

\section{Conclusion}\label{sec:conclusion}

In this work, we propose an adaptive black-box defense against Trojan attacks, dubbed \textbf{\textsc{TrojDef}}. \textbf{\textsc{TrojDef}} perturbs each input example with random Gaussian noise and utilizes the prediction of the perturbed examples to decide whether the input example contains the Trojan trigger or not. We show analytically that under restricted conditions \textbf{\textsc{TrojDef}} can always differentiate benign from Trojan examples by deriving prediction confidence bound. We also propose a non-linear transformation to the prediction confidence bound to enable accurate detection of Trojan examples when the restricted conditions do not hold.  We also propose several practical enhancements to \textbf{\textsc{TrojDef}}, especially when the input examples are images. We conduct several experiments to compare \textbf{\textsc{TrojDef}} with the SOTA black-box approach, STRIP. The results show that \textbf{\textsc{TrojDef}} has a competitive performance on all the experiments proposed by STRIP. Moreover, the results in the expanded experiments show that \textbf{\textsc{TrojDef}} not only outperforms STRIP but is also more stable. The performance of STRIP may significantly degenerate when (1) the NN classifiers' training hyper-parameters change or (2) the NN classifier's architecture changes. Under similar settings, \textbf{\textsc{TrojDef}} provides consistent performance. In addition, we evaluate \textbf{\textsc{TrojDef}} and STRIP on a more realistic scenario when the Trojan backdoor is implanted in a large-scale NN classifier pre-trained on other datasets. The results show that \textbf{\textsc{TrojDef}} significantly outperforms STRIP under such challenging settings. Finally, by replacing the Gaussian perturbation with Laplacian ones, the results confirm the generalizability of the \textbf{\textsc{TrojDef}} to arbitrary datasets and arbitrary noise distributions. The main reason for this superior performance is that \textbf{\textsc{TrojDef}} is controllable and can easily adapt to the presented examples by changing the parameters of the distribution of the added random noise.

\section{Limitations and Future Work} \label{sec:future}

\guanxiongVtwo{
Based on Section \ref{sec:defense}, it is not hard to imagine that if the prediction on Trojan example is sensitive towards the added noise, the performance of \textbf{\textsc{TrojDef}} will be degenerated. We observe this degeneration when evaluating \textbf{\textsc{TrojDef}} against the Hello Kitty pattern trigger presented in \cite{chen2017targeted} with 90\% transparency. The results are summarized in Table \ref{table:hello-kitty}. Although preparing this Trojan attack require the attacker to perturb the entire image which is more visible in human eyes, we think there are some interesting problems that are worth studying in the future.
}

\begin{table}[h]
    \scriptsize
    \begin{center}
    
    \begin{tabular}{ c  c | c  c  c }
    \hline \hline
    Dataset  &
    Model &
    Acc &
    Attack-Acc &
    FAR \\
     \hline
    \multirow{2}{*}{CIFAR-10}  
    &STRIP & 78.57\% & 68.67\% &   91.4\%   \\
    &   \textbf{\textsc{TrojDef}}& 78.5\%& 66.94\% & 95.75\%   \\
   
  \hline \hline
    \end{tabular}
    
    \end{center}
    \caption{Results of pattern trigger at 90\% transparent}
    \label{table:hello-kitty}
\end{table}

\begin{enumerate}
    
    
    \item \guanxiongVtwo{Even when the prediction of Trojan examples are sensitive towards perturbation, we believe it is different from the benign examples due to the difference in extracted features. To distinguish invisible Trojan examples, the method of generating adversarial perturbation could be utilized. Also, to keep it as a black-box defense, we can focus on methods that only utilize zero-th order gradient information when generating adversarial perturbation \cite{chen2017zoo, liu2020manigen, ughi2020empirical}}
    
    \item In Theorem~\ref{th:delta-diff}, we can see that the optimal way of adding the perturbations is to make them correlated to the distribution of the training data. Although we have demonstrated in this paper that a decent performance can be achieved when we ignore the knowledge about the training data, such knowledge might be available under some practical scenarios. In our future work we will identify these scenarios and decide how to perform the actual correlation between the knowledge of the training data and the exact way to add the perturbation.
\end{enumerate}

\bibliographystyle{IEEEtran}
\bibliography{reference.bib}
\newpage
\onecolumn
\appendix

\begin{table*}[h!]
    \scriptsize
    \parbox{.45\linewidth}{
        \centering
        \begin{tabular}{ c  c | c  c  c  c }
            \hline \hline
            \multirow{2}{*}{Dataset} & \multirow{2}{*}{Trigger} &
            \multirow{2}{*}{Model} &
            \multicolumn{2}{c}{FAR} \\
            & & &\multirow{1}{*}{B\tnote{1}} & \multirow{1}{*}{RGB \tnote{2}} \\ 
            \hline
            \multirow{2}{*}{CIFAR-10}  
            &\multirow{3}{*}{"face"} &  STRIP & 0.0\% &   0.0\% \\
            & & \textbf{\textsc{TrojDef}} & 0.0\%&0.0\% \\
            & & 3rd-party & 0.0\% & 0.0\% \\
            & \multirow{3}{*}{"watermark"} &  STRIP& 0.0\% & 0.0\% \\
             & &  \textbf{\textsc{TrojDef}}& 0.0\% & 0.0\% \\
              &  &  3rd-party& 0.0\% & 0.0\% \\
              & \multirow{2}{*}{"Bottle"} &    \textbf{\textsc{TrojDef}}& 0.15\% & 0.15\% \\
              &  &  3rd-party& 19.5\% & 34.3\% \\
               & \multirow{2}{*}{"Star"} &    \textbf{\textsc{TrojDef}}& 0.0\% & 0.0\% \\
              &  &  3rd-party& 0.0\% & 0.0\% \\
            \hline
             \multirow{2}{*}{GTSRB}
            &\multirow{3}{*}{"face"} &  STRIP & 0.0\% &   0.0\% \\
            & & \textbf{\textsc{TrojDef}} & 0.0\%&0.0\% \\
            & & 3rd-party & 0.0\% & 0.0\% \\
            & \multirow{3}{*}{"watermark"}&  \textbf{\textsc{TrojDef}}& 0.0\% & 0.0\% \\
              &  &  3rd-party& 0.0\% & 100.0\% \\
               & \multirow{2}{*}{"Bottle"} &    \textbf{\textsc{TrojDef}}& 0.05\% & 0.05\% \\
              &  &  3rd-party& 0.05\% & 0.05\% \\
               & \multirow{2}{*}{"Star"} &    \textbf{\textsc{TrojDef}}& 0.0\% & 0.0\% \\
              &  &  3rd-party& 0.0\% & 45.7\% \\
            \hline
             \multirow{2}{*}{CUB200}
            &"face" &  \multirow{2}{*}{3rd-party} & 1.15\% &   100.0\% \\
            &"watermark" &  & 0.0\%&100.0\% \\
            \hline \hline
        \end{tabular}
        \begin{tablenotes}
            \item [1] B:Blue channel.
           \item [2] RGB: All channels.
        \end{tablenotes}
        \caption{Results of applying perturbation on blue channel and all channels experiments}
        \label{table:all-channel}
    }
    \hfill
    \parbox{.45\linewidth}{
        \centering
        \begin{tabular}{ c  c | c  c  c  c }
            \hline \hline
            \multirow{2}{*}{Dataset} & \multirow{2}{*}{Trigger} &
            \multirow{2}{*}{Model} &
            \multicolumn{3}{c}{FAR} \\
            & & &\multirow{1}{*}{B\tnote{1} } & \multirow{1}{*}{R\tnote{2} }
            & \multirow{1}{*}{G\tnote{3} }\\ 
            \hline
            \multirow{2}{*}{CIFAR-10}  
            &\multirow{3}{*}{"face"} &  STRIP & 0.0\% &   0.0\% & 0.0\% \\
            & & \textbf{\textsc{TrojDef}} & 0.0\% & 0.0\% & 0.0\% \\
            & & 3rd-party & 0.0\% & 0.0\% & 0.0\% \\
            & \multirow{3}{*}{"watermark"} &  STRIP& 0.0\% & 0.0\% & 0.0\% \\
             & &  \textbf{\textsc{TrojDef}}& 0.0\% & 0.0\% & 0.0\% \\
              &  &  3rd-party& 0.0\% & 0.0\% & 0.0\% \\
              & \multirow{2}{*}{"Bottle"} &    \textbf{\textsc{TrojDef}}& 0.15\% & 0.15\%  & 0.15\% \\
              &  &  3rd-party& 19.5\% & 100.0\% & 19.90\% \\
              & \multirow{2}{*}{"Star"} &    \textbf{\textsc{TrojDef}}& 0.0\% & 0.0\% & 0.0\% \\
              &  &  3rd-party& 0.0\% & 0.0\% & 0.0\% \\
            \hline
             \multirow{2}{*}{GTSRB}
            &\multirow{3}{*}{"face"} &  STRIP & 0.0\% &   100.0\% &   100.0\% \\
            & & \textbf{\textsc{TrojDef}} & 0.0\%&100.0\% &   100.0\%\\
            & & 3rd-party & 0.0\% & 100.0\% &   100.0\% \\
            & \multirow{3}{*}{"watermark"}&  \textbf{\textsc{TrojDef}}& 0.0\% & 100.0\% & 100.0\% \\
              &  &  3rd-party& 0.0\% & 100.0\% & 100.0\%  \\
              & \multirow{2}{*}{"Bottle"} &    \textbf{\textsc{TrojDef}}& 0.05\% & 100.0\% &   100.0\%\\
              &  &  3rd-party& 0.05\% & 100.0\% &   100.0\%\\
              & \multirow{2}{*}{"Star"} &    \textbf{\textsc{TrojDef}}& 0.0\% & 100.0\% & 0.0\%\\
              &  &  3rd-party& 0.0\% & 100.0\% & 0.0\% \\
            \hline
             \multirow{2}{*}{CUB200}
            &"face" &  \multirow{2}{*}{3rd-party} & 1.15\% &   100.0\% & 100.0\% \\
            &"watermark" &  & 0.0\%&100.0\% & 100.0\% \\
            \hline \hline
        \end{tabular}
        \begin{tablenotes}
          \item [1] B: Blue channel.
          \item [2] R: red channel.
          \item [2] G: green channel.
        \end{tablenotes}
        \caption{Results of applying perturbation on different single-channel  experiments}
        \label{table:blue-vs-other}
    }
        \hfill
    \parbox{.45\linewidth}{
        \centering
        \begin{tabular}{ c  c | c  c   c }
            \hline \hline
            \multirow{2}{*}{Dataset} & \multirow{2}{*}{Trigger} &
            
            \multicolumn{3}{c}{FAR} \\
            & & B\tnote{1}  & R\tnote{2} 
            & G\tnote{3} \\ 
            \hline
             CUB200
            &"blue star"  & 1.30\% &   1.20\% & 9.50\% \\

            \hline \hline
        \end{tabular}
        \begin{tablenotes}
          \item [1] B: Blue channel.
          \item [2] R: red channel.
          \item [2] G: green channel.
        \end{tablenotes}
        \caption{Results of adding trigger  on different channels  with cub200 dataset experiments}
        \label{table:cub200_rgb}
    }
\end{table*}

\begin{table*}
    \scriptsize
    \parbox{.45\linewidth}{
        \centering
        \begin{tabular}{ c  c | c  c  c  c c }
            \hline \hline
            \multirow{3}{*}{Dataset} & \multirow{3}{*}{Trigger} &
            \multirow{3}{*}{Model} &
            \multicolumn{4}{c}{FAR} \\
             & & &\multicolumn{2}{c}{With\tnote{1}}&
            \multicolumn{2}{c}{Without\tnote{2}} \\
             & & & B\tnote{3} & RGB\tnote{4} &B & RGB \\
            
            \hline
            \multirow{2}{*}{CIFAR-10}  
            &\multirow{3}{*}{"face"} &  STRIP & 0.0\% &   0.0\% & 0.0\%  & 100.0\%  \\
            & & \textbf{\textsc{TrojDef}} & 0.0\% & 0.0\% & 100.0\%  & 100.0\%  \\
            & & 3rd-party & 0.0\% & 0.0\% & 100.0\%  & 100.0\%  \\
            & \multirow{3}{*}{"watermark"} &  STRIP& 0.0\% & 0.0\% & 100.0\%  & 100.0\%  \\
             & &  \textbf{\textsc{TrojDef}}& 0.0\% & 0.0\% & 0.0\%   & 100.0\%  \\
              &  &  3rd-party& 0.0\% & 0.0\% & 100.0\%  & 100.0\%  \\
              & \multirow{2}{*}{"Bottle"} &    \textbf{\textsc{TrojDef}}& 0.15\% & 0.15\%  & 100.0\%  & 100.0\%  \\
              &  &  3rd-party& 19.5\% & 34.3\% & 100.0\%  & 100.0\% \\
              & \multirow{2}{*}{"Star"} &    \textbf{\textsc{TrojDef}}& 0.0\% & 0.0\% & 0.0\%  & 100.0\%  \\
              &  &  3rd-party& 0.0\% & 0.0\% & 100.0\%  & 100.0\%  \\
            \hline
             \multirow{2}{*}{GTSRB}
            &\multirow{3}{*}{"face"} &  STRIP & 0.0\% &   100.0\% &   100.0\% & 100.0\%   \\
            & & \textbf{\textsc{TrojDef}} & 0.0\%&100.0\% &   100.0\%  & 100.0\%  \\
            & & 3rd-party & 0.0\% & 100.0\% &   100.0\%  & 100.0\%  \\
            & \multirow{3}{*}{"watermark"}&  \textbf{\textsc{TrojDef}}& 0.0\% & 100.0\% & 100.0\%  & 100.0\% \\
              &  &  3rd-party& 0.0\% & 100.0\% & 100.0\%  & 100.0\%  \\
              & \multirow{2}{*}{"Bottle"} &    \textbf{\textsc{TrojDef}}& 0.05\% & 100.0\% &   100.0\%  & 100.0\% \\
              &  &  3rd-party& 0.05\% & 100.0\% &   100.0\%  & 100.0\%  \\
              & \multirow{2}{*}{"Star"} &    \textbf{\textsc{TrojDef}}& 0.0\% & 100.0\% & 100.0\%  & 100.0\%  \\
              &  &  3rd-part& 0.0\% & 45.7\% & 100.0\%  & 100.0\% \\
            \hline
             \multirow{2}{*}{CUB200}
            &"face" &  \multirow{2}{*}{3rd-party} & 1.15\% &   100.0\% & 100.0\%  & 100.0\%  \\
            &"watermark" &  & 0.0\%&100.0\% & 100.0\%  & 100.0\%  \\
          \hline \hline
        \end{tabular}
        \begin{tablenotes}
            \item [1] With random size and location.
            \item [2] Without random size and location.
            \item [3] B: Blue Channel.
            \item [4] RGB: All channels.
        \end{tablenotes}
        \caption{Results of the experiments with/without random size and location on blue channel and all channels}
        \label{table:rand-size-exp}
    }
    \hfill
    \parbox{.45\linewidth}{
        \centering
        \begin{tabular}{ c  c | c  c  c   }
            \hline \hline
            \multirow{3}{*}{Dataset} & \multirow{3}{*}{Trigger} &
            \multirow{3}{*}{Model} &
            \multicolumn{2}{c}{FAR} \\
             & & & B\tnote{1}/With\tnote{2} &
           B/Without\tnote{3} \\
            
            \hline
            \multirow{2}{*}{CIFAR-10}  
            &\multirow{3}{*}{"face"} &  STRIP & 0.0\% &   0.0\%   \\
            & & \textbf{\textsc{TrojDef}} & 0.0\% & 0.0\%  \\
            & & 3rd-party & 0.0\% & 0.0\% \\
            & \multirow{3}{*}{"watermark"} &  STRIP& 0.0\% & 0.0\%  \\
             & &  \textbf{\textsc{TrojDef}}& 0.0\% & 0.0\%  \\
              &  &  3rd-party& 0.0\% & 0.0\%   \\
              & \multirow{2}{*}{"Bottle"} &    \textbf{\textsc{TrojDef}}& 0.15\% & 0.15\%    \\
              &  &  3rd-party& 19.5\% & 61.75\%  \\
               & \multirow{2}{*}{"Star"} &    \textbf{\textsc{TrojDef}}& 0.0\% & 100.0\%   \\
              &  &  3rd-party& 0.0\% & 0.0\% \\
            \hline
             \multirow{2}{*}{GTSRB}
            &\multirow{3}{*}{"face"} &  STRIP & 0.0\% &   100.0\%  \\
            & & \textbf{\textsc{TrojDef}} & 0.0\%&100.0\% \\
            & & 3rd-party & 0.0\% & 100.0\%  \\
            & \multirow{3}{*}{"watermark"}&  \textbf{\textsc{TrojDef}}& 0.0\% & 100.0\% \\
              &  &  3rd-party& 0.0\% & 100.0\%  \\
               & \multirow{2}{*}{"Bottle"} &    \textbf{\textsc{TrojDef}}& 0.05\% & 100.0\%  \\
              &  &  3rd-party& 0.05\% & 100.0\% \\
               & \multirow{2}{*}{"Star"} &    \textbf{\textsc{TrojDef}}& 0.0\% & 100.0\%   \\
              &  &  3rd-party& 0.0\% & 100.0\% \\
            \hline
             \multirow{2}{*}{CUB200}
            &"face" &  \multirow{2}{*}{3rd-party} & 1.15\% &   100.0\%   \\
            &"watermark" &  & 0.0\%&100.0\%   \\
          \hline \hline
        \end{tabular}
        \begin{tablenotes}
            \item [1] With Dynamic Standard Deviation.
            \item [2] Without Dynamic Standard Deviation.
            \item [3] B: Blue Channel.
        \end{tablenotes}
        \caption{Results of the experiments with/without dynamic standard deviation on blue channel}
        \label{table:dyn-sigma}
    }
\end{table*}

\begin{table*}[h]
    \scriptsize
    \begin{center}
     \begin{threeparttable}[b]
    \begin{tabular}{ c  c | c  c  c  c  }
    \hline \hline
    \multirow{3}{*}{Dataset} & \multirow{3}{*}{Trigger} &
    \multirow{3}{*}{Model} &
    \multicolumn{3}{c}{FAR} \\
     & & & \multirow{2}{*}{All enhancements}&
    \multicolumn{2}{c}{Without enhancements on random noise \tnote{1}} \\
    & & & & B\tnote{2}& RGB \tnote{3} \\

    \hline
    \multirow{2}{*}{CIFAR-10}  
    &\multirow{3}{*}{"face"} &  STRIP & 0.0\% &   100.0\% & 100.0\%   \\
    & & \textbf{\textsc{TrojDef}} & 0.0\% & 100.0\% & 100.0\%   \\
    & & 3rd-party & 0.0\% & 100.0\% & 100.0\%   \\
    & \multirow{3}{*}{"watermark"} &  STRIP& 0.0\% & 100.0\% & 100.0\%   \\
     & &  \textbf{\textsc{TrojDef}}& 0.0\% & 100.0\% & 0.0\%   \\
      &  &  3rd-party& 0.0\% & 100.0\% & 100.0\%    \\
      & \multirow{2}{*}{"Bottle"} &    \textbf{\textsc{TrojDef}}& 0.15\% & 100.0\%  & 100.0\%   \\
      &  &  3rd-party& 19.5\% & 100.0\% & 100.0\%  \\
       & \multirow{2}{*}{"Star"} &    \textbf{\textsc{TrojDef}}& 0.0\% & 100.0\% & 100.0\%   \\
      &  &  3rd-party& 0.0\% & 100.0\% & 100.0\%   \\
    \hline
     \multirow{2}{*}{GTSRB}
    &\multirow{3}{*}{"face"} &  STRIP & 0.0\% &   100.0\% &   100.0\%   \\
    & & \textbf{\textsc{TrojDef}} & 0.0\%&100.0\% &   100.0\%    \\
    & & 3rd-party & 100.0\% & 100.0\% &   100.0\%    \\
    & \multirow{3}{*}{"watermark"}&  \textbf{\textsc{TrojDef}}& 0.0\% & 100.0\% & 100.0\%  \\
      &  &  3rd-party& 0.0\% & 100.0\% & 100.0\%    \\
       & \multirow{2}{*}{"Bottle"} &    \textbf{\textsc{TrojDef}}& 0.05\% & 100.0\% &   100.0\%  \\
      &  &  3rd-party& 0.05\% & 100.0\% &   100.0\%    \\
       & \multirow{2}{*}{"Star"} &    \textbf{\textsc{TrojDef}}& 0.0\% & 100.0\% & 100.0\%   \\
      &  &  3rd-part& 0.0\% & 100.0\% & 100.0\%   \\
    \hline
     \multirow{2}{*}{CUB200}
    &"face" &  \multirow{2}{*}{3rd-party} & 1.15\% &   100.0\% & 100.0\%    \\
    &"watermark" &  & 0.0\%&100.0\% & 100.0\%    \\
  \hline \hline
    \end{tabular}
    \begin{tablenotes}
        \item [1] Randomizing the size \& location of the Gaussian perturbation and utilizing dynamic standard deviation.
       \item [2] B: Blue Channel.
       \item [3] RGB: All channels.
     \end{tablenotes}
     \end{threeparttable}
    \end{center}
    \caption{The results without any enhancement compared with all enhancements}
    \label{table:no-enhance}
\end{table*}

\begin{figure}[htb]
    \centering
   
    \begin{subfigure}[b]{\textwidth}
        \includegraphics[width=\textwidth]{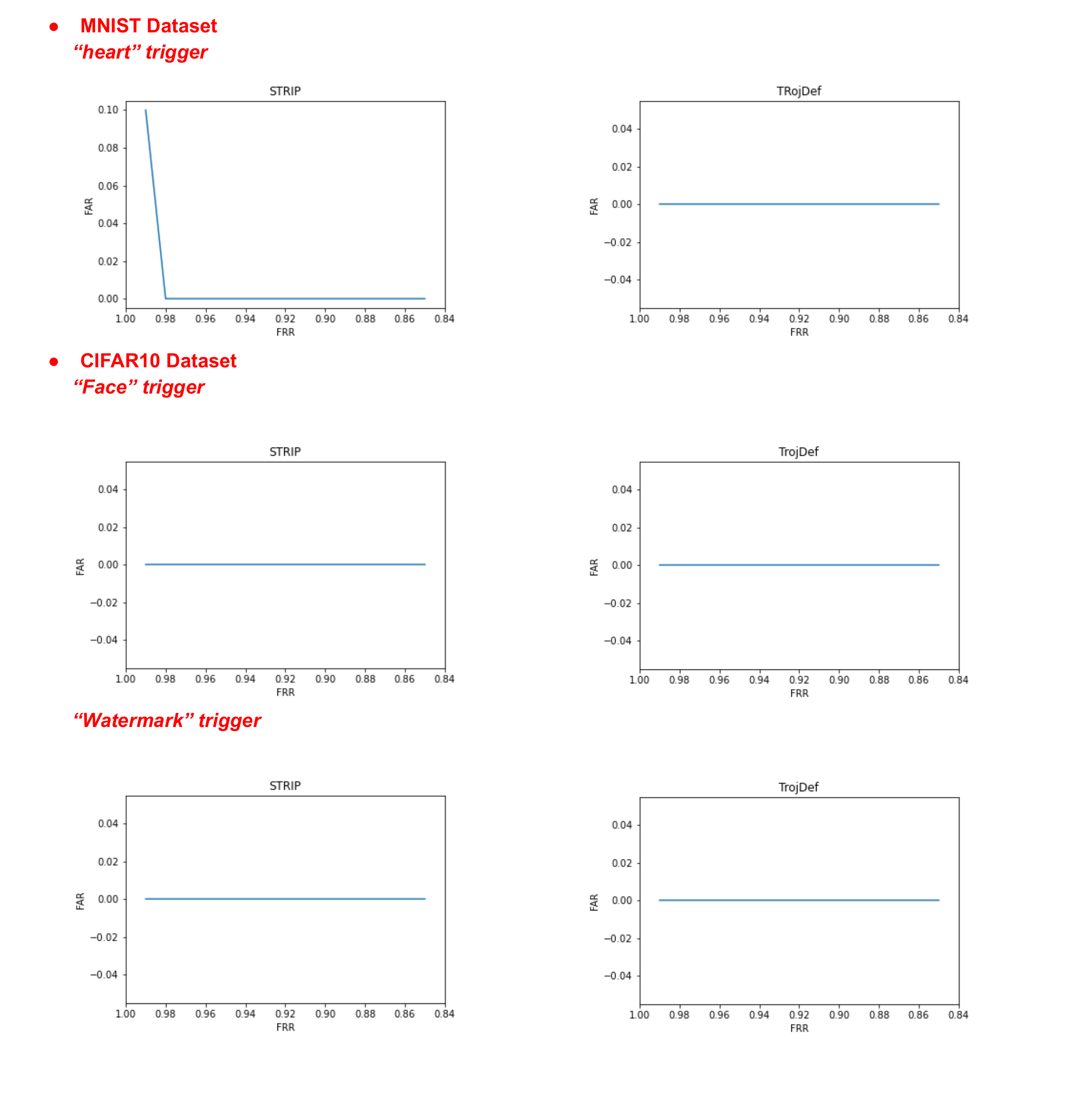}
       
    \end{subfigure}
    \caption{Relationship between FRR and FAR for the experiments with STRIP-model}
    \label{fig:strip_model_ROC}
\end{figure}
\begin{figure}[htb]
    \centering
    \begin{subfigure}[b]{\textwidth}
        \includegraphics[page=1,width=\textwidth]{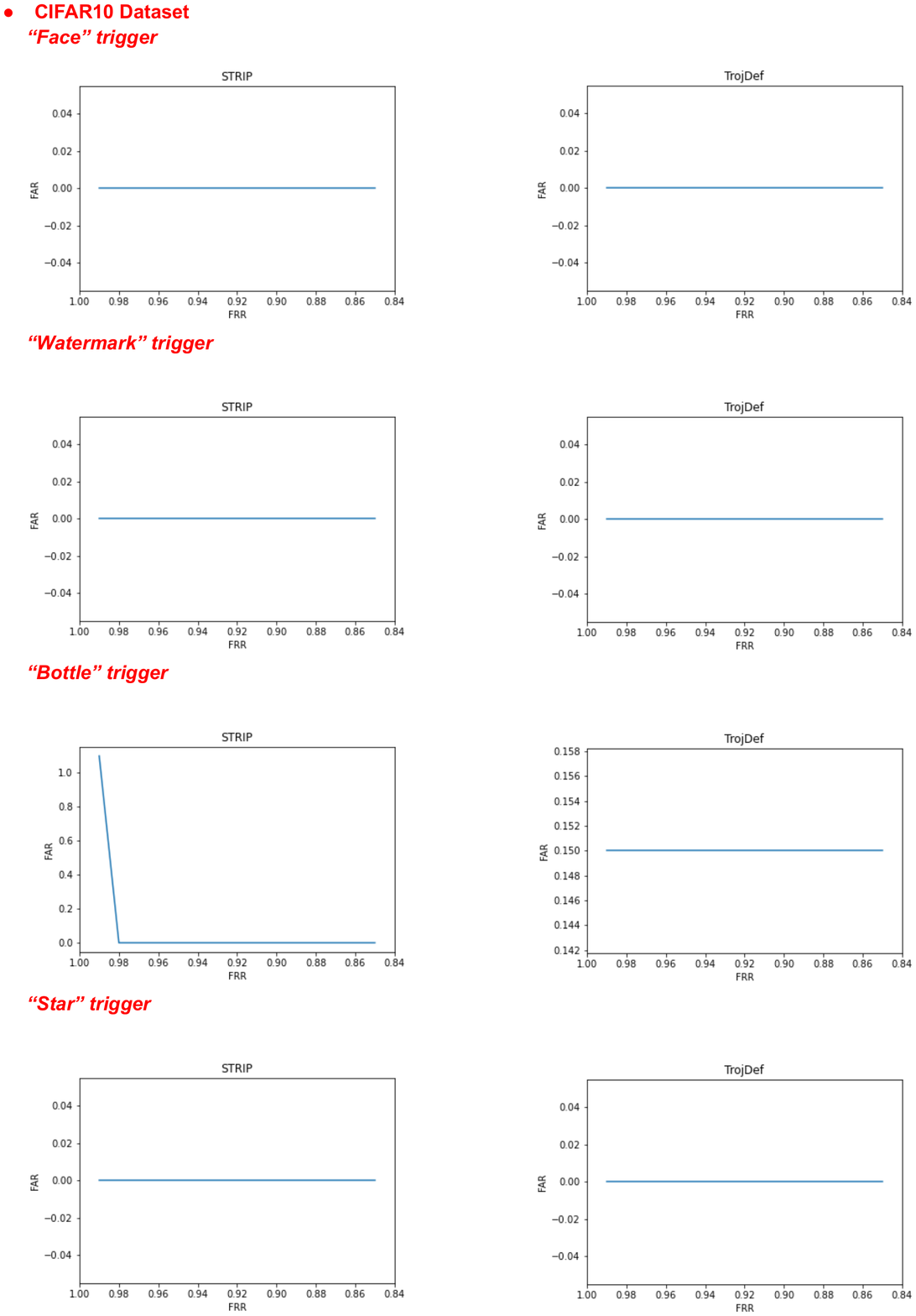}
       
    \end{subfigure}

    \caption{Relationship between FRR and FAR for the experiments with \textbf{\textsc{TrojDef}}-model}
   
\end{figure}
\begin{figure}[htb]\ContinuedFloat
    \centering
   
    \begin{subfigure}[b]{\textwidth}
        \includegraphics[page=2,width=\textwidth]{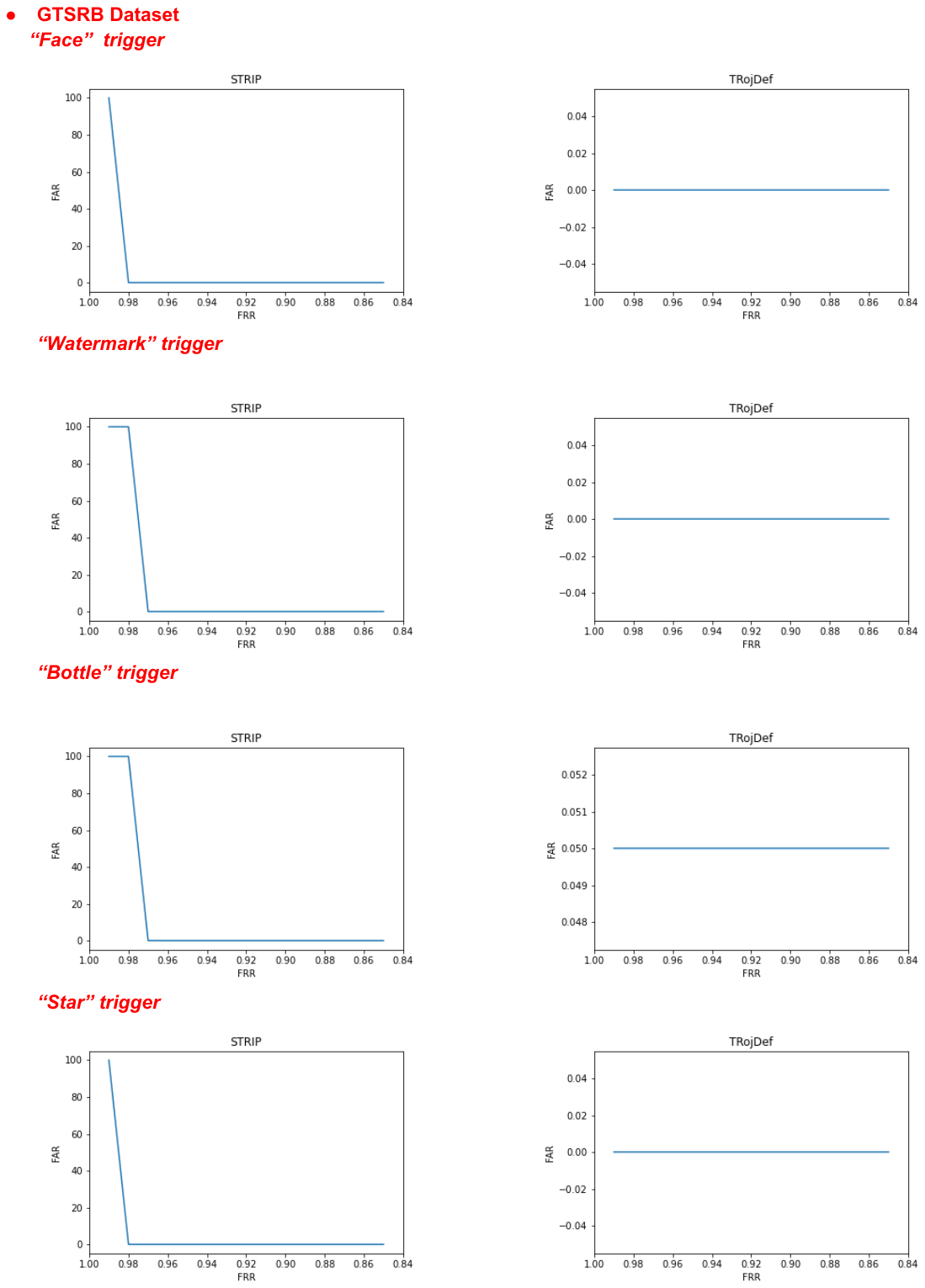}
       
    \end{subfigure}
    \caption{Relationship between FRR and FAR for the experiments with \textbf{\textsc{TrojDef}}-model (cont.)}
    \label{fig:our_model_ROC}
\end{figure}
\begin{figure}[htb]
    \centering
    \begin{subfigure}[b]{\textwidth}
        \includegraphics[page=1,width=\textwidth]{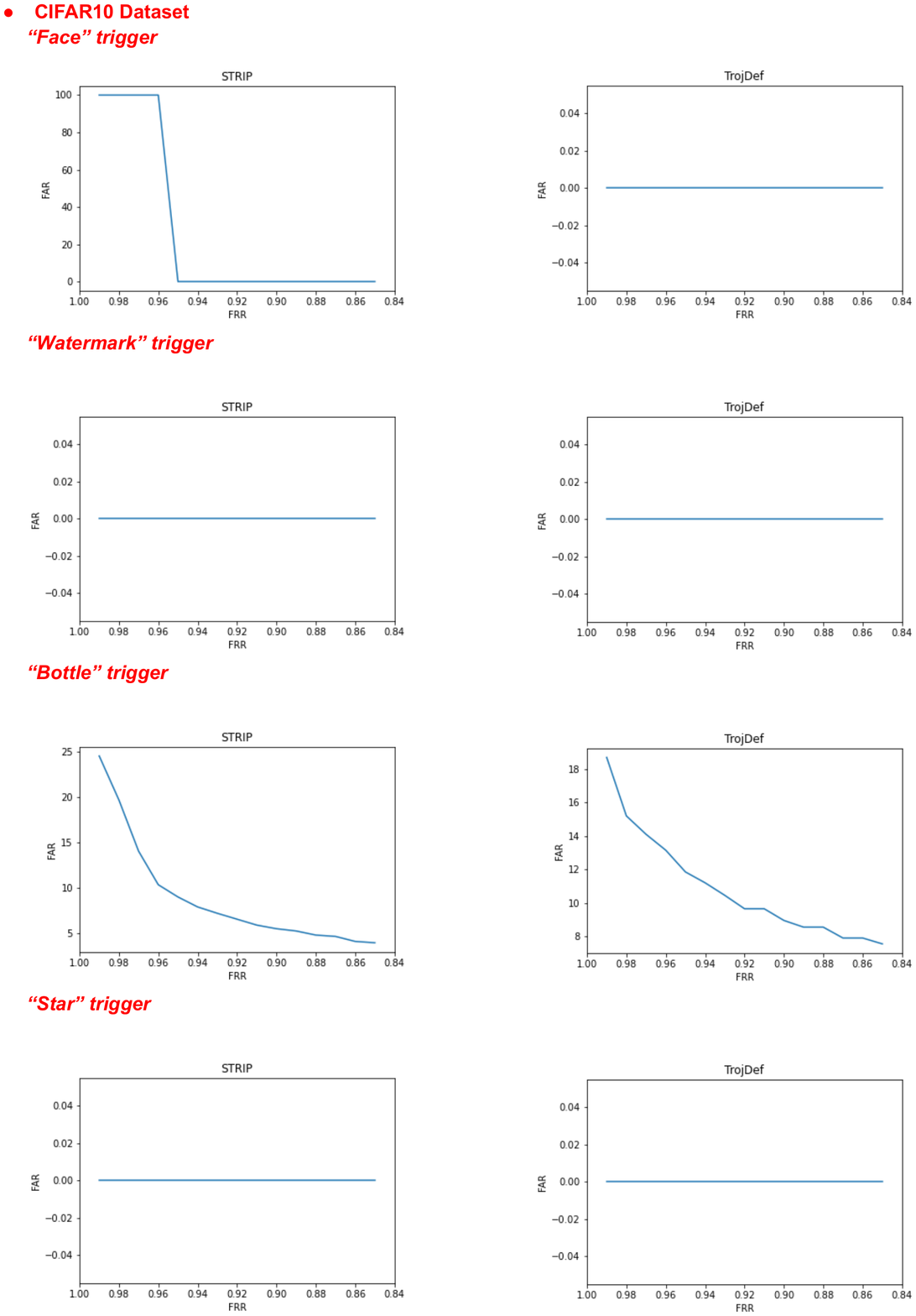}
       
    \end{subfigure}

    \caption{Relationship between FRR and FAR for the experiments with 3rd-party model }
   
\end{figure}
\begin{figure}[htb]\ContinuedFloat
    \centering
    \begin{subfigure}[b]{\textwidth}
        \includegraphics[page=2,width=\textwidth]{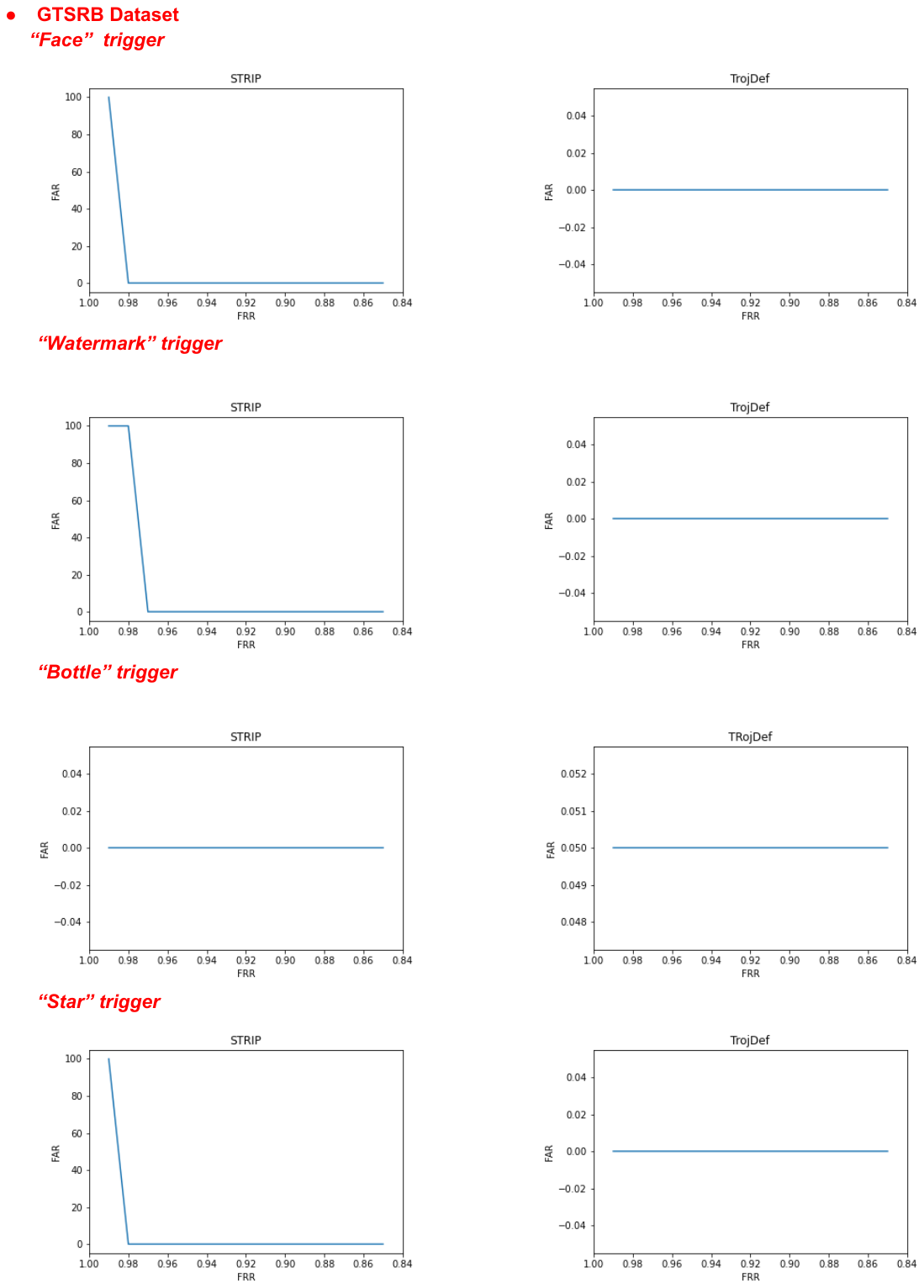}
       
    \end{subfigure}

    \caption{Relationship between FRR and FAR for the experiments with 3rd-party model (cont.)}
   
\end{figure}

\begin{figure}[htb]\ContinuedFloat
    \centering
   
    \begin{subfigure}[b]{\textwidth}
        \includegraphics[page=3,width=\textwidth]{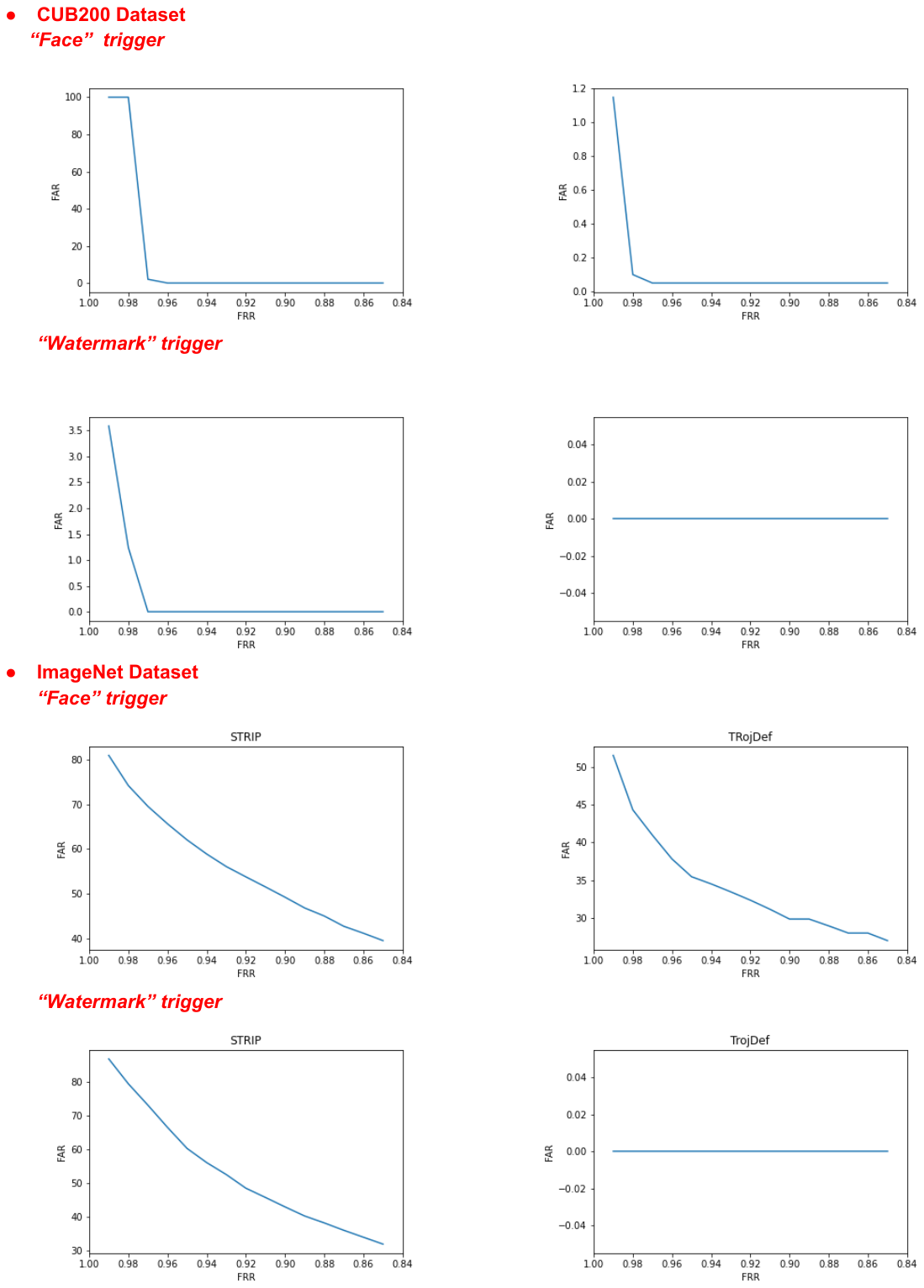}
       
    \end{subfigure}
    \caption{Relationship between FRR and FAR for the experiments with 3rd-party model (cont.)}
    \label{fig:3rd_model_ROC}
\end{figure}

\end{document}